\documentclass[letterpaper,twocolumn,superscriptaddress,
               10pt, unpublished]{quantumarticle}
\pdfoutput=1

\usepackage[utf8]{inputenc}
\usepackage[english]{babel}
\usepackage[T1]{fontenc}

\usepackage{lipsum}
\usepackage{physics}
\usepackage{amsmath}
\usepackage{amsthm}
\usepackage{amsfonts, amssymb}
\usepackage{mathtools}
\usepackage{blkarray}
\usepackage{caption, subcaption}

\usepackage{tikz}
\usepackage{quantikz}
\usepackage{float}

\usepackage{hyperref}
\newtheorem{definition}{Definition}
\newtheorem{theorem}{Theorem}[section]
\newtheorem{proposition}{Proposition}[section]
\newtheorem{corollary}{Corollary}[theorem]
\newtheorem{lemma}[theorem]{Lemma}

\newcommand{\s}{{s}}
\begin{document}
\title{Emergent Non-Markovianity in Logical Qubit Dynamics}

\author{Jalan A. Ziyad}
\email[]{jaziyad@sandia.gov}
\affiliation{Quantum Performance Laboratory, Sandia National Laboratories, Albuquerque, NM, 87185, USA}
\affiliation{Center for Quantum Information and Control, University of New Mexico, Albuquerque, NM, 87131, USA}
\affiliation{Department of Physics and Astronomy, University of New Mexico, Albuquerque, NM, 87131, USA}

\author{Robin Blume-Kohout}
%\email[]{rjblume@sandia.gov}
\affiliation{Quantum Performance Laboratory, Sandia National Laboratories, Albuquerque, NM, 87185, USA}
\affiliation{Center for Quantum Information and Control, University of New Mexico, Albuquerque, NM, 87131, USA}

\author{Kenneth Rudinger}
%\email[]{kmrudin@sandia.gov}
\affiliation{Quantum Performance Laboratory, Sandia National Laboratories, Albuquerque, NM, 87185, USA}

\date{December 8, 2025}

\begin{abstract}
    \noindent Logical qubits encoded in quantum error correcting codes can exhibit non-Markovian dynamical evolution, even when the underlying physical noise is Markovian.  To understand this emergent non-Markovianity, we define a Markovianity condition appropriate to logical gate operations, and study it by relating logical operations to their physical implementation (operations on the data qubits into which the logical qubit is encoded).  We apply our analysis to small quantum codes, and show that they exhibit non-Markovian dynamics even for very simple physical noise models.  We show that non-Markovianity can emerge from Markovian physical operations if (and only if) the physical qubits are not necessarily returned to the code subspace after every round of QEC.  In this situation, the syndrome qubits can act as a memory, mediating time correlations and enabling violation of the Markov condition.  We quantify the emergent non-Markovianity in simple examples, and propose sufficient conditions for reliable use of gate-based characterization techniques like gate set tomography in early fault-tolerant quantum devices.
\end{abstract} 

\maketitle

\section{Introduction}

Fault-tolerant quantum computing (FTQC) relies upon quantum error correction (QEC) \cite{knill_theory_1997,knill_theory_2000, shor1995scheme, shor1996fault, knill_resilient_1998,knill1996threshold, aleferis_quantum_2006, kitaev2003fault, aharonov1997fault, dennis_topological_2002, fowler2012surface}, which encodes quantum information across many physical qubits to protect it from environmental and control noise.  The failure probability or \emph{error rate} of a \emph{logical qubit} protected by QEC may be far lower than that of  any unprotected physical qubit, but it is not zero.  Logical qubits still experience errors.  The error processes of logical qubits can be studied, modeled, and characterized in experiments.  Understanding and modeling logical error processes is on the critical path to fault-tolerant quantum computing, because FTQC systems must be designed and then experimentally validated to achieve specific (low) failure probabilities.  One of the most fundamental questions about any qubit, physical or logical, is ``Are its error processes \emph{Markovian}?'' -- i.e., can a modeler confidently expect that faults at distinct times will occur independently?  The answer determines whether or not a wide range of techniques for \emph{quantum characterization verification and validation} (QCVV) \cite{blume2025quantum,hashim2408practical,proctor2025benchmarking} can be used reliably to characterize the qubit's behavior.

In principle, a logical qubit could be studied and modeled by characterizing its component physical qubits using QCVV, then simulating their detailed behavior in QEC and other FTQC primitives.  But unless the QEC code is very small or the physical qubits' noise has a simple structure \cite{iyer_small_2018, iverson_coherence_2020, rahn_exact_2002, beale_efficiently_2021, bravyi2018correcting}, this is often infeasible.  The simulation may require extraordinary computational resources, and characterizing the physical-qubit operations in sufficient detail (including, e.g., crosstalk effects) may require extraordinary experimental resources.  A simpler and appealing alternative is to probe the performance of logical qubits directly by adapting QCVV techniques used for physical qubits, such as gate set tomography (GST) \cite{greenbaum_introduction_2015,Nielsen2021gatesettomography} or randomized benchmarking (RB) \cite{emerson_scalable_2005, proctor_what_2017}, to characterize error processes in logical qubits directly \cite{rudinger2023probing,combes2017logical}.

QCVV protocols depend on assumptions about the error processes they characterize, and the validity of those assumptions controls a protocol's reliability when applied (e.g., to logical qubits).  Most QCVV protocols assume that every time a gate is applied, its action is well-described by a particular fixed completely positive trace-preserving (CPTP) map -- regardless of the gate's context (e.g., the time of day, or what gates were applied previously). This assumption is informally known as \emph{Markovianity}, by analogy to the Markov condition for classical stochastic processes \cite{levin2017markov}.  It is broadly understood that physical qubits can display non-Markovian behavior \cite{burkard2009non,white2020demonstration,rudinger_probing_2019}. It stands to reason that if a logical qubit is built from sufficiently non-Markovian physical qubits, then the logical qubit too can display non-Markovian behavior.  In this paper, we investigate the reverse implication:  if a logical qubit's component physical qubits \textit{are} perfectly described by Markovian models, does that imply that the logical qubit's error processes will \emph{also} be Markovian?

We find that the answer is ``no''. We define a precise condition for a qubit's operations (gates) to be Markovian, and then show that even when this Markovianity condition holds for all the physical qubits making up a logical qubit, the failure probability of a \emph{logical} operation can depend on what happened in the past. This is a non-Markovian effect, and it produces a clear signature as non-exponential decay \cite{ceasura_non-exponential_2022} of the logical qubit's polarization, $\langle \Bar{Z} \rangle$.  Investigating further, we find and prove that this emergent non-Markovianity occurs only if the physical data qubits encoding the logical qubit do not return to the code subspace after every round of QEC.  When this condition is satisfied, the overall system retains persistent information outside of the logical qubit subsystem.  The data qubits' Hilbert space can be factored into a logical qubit subsystem and a set of \emph{syndrome qubits}, which act as a persistent environment for the logical qubit that contains memory of the recent past and feeds it back into the logical dynamics to produce detectable non-Markovianity.

Previous work investigated several related topics. It has been periodically observed in the literature \cite{knill_theory_1997,knill_theory_2000, carrozza2024correspondence} that logical qubits can be treated as subsystems, and the study of effective logical maps was initiated in Refs.~\cite{iverson_coherence_2020, rahn_exact_2002, beale_efficiently_2021, bravyi2018correcting}. Several recent papers have sought to understand logical dynamics in the context of non-Markovian noise \cite{tanggara2024strategic, kam2024detrimental, kobayashi2024tensor}.  In particular, Ref.~\cite{ceasura_non-exponential_2022} pointed out that logical non-Markovianity can emerge from an underlying Markovian model and can impact logical randomized benchmarking. Our work extends this observation by giving a precise definition of logical Markovianity that is independent of the theory of RB, and exploring explicit examples of logical non-Markovianity in specific syndrome extraction circuits. We also provide an argument that Markovianity will generically be violated (albeit perhaps by only a small amount) in realistic error correction scenarios, and a conceptual foundation for studying emergent non-Markovianity with general noise models. While preparing this paper for publication we became aware of related work by Endo \textit{et al} \cite{endo2025non}, which also addresses non-Markovianity in QEC but with a different focus (error mitigation) and concept of non-Markovianity (CP-divisibility).

We proceed as follows.  In Section \ref{sec: NM}, we lay out a precise definition of logical non-Markovianity. Then, in Section \ref{sec: example}, we define and study a toy model that demonstrates the existence of emergent non-Markovian logical processes. We construct a more general theory explaining \emph{why} this occurs, rooted in the idea that logical qubits are virtual subsystems of the physical data qubits, in Section \ref{sec: causes}.  We conclude with a summary of our paper's key implications and a brief survey of pressing open questions.

\section{Logical non-Markovianity}\label{sec: NM}

\begin{figure}
    \centering
    \includegraphics[width=0.45\linewidth]{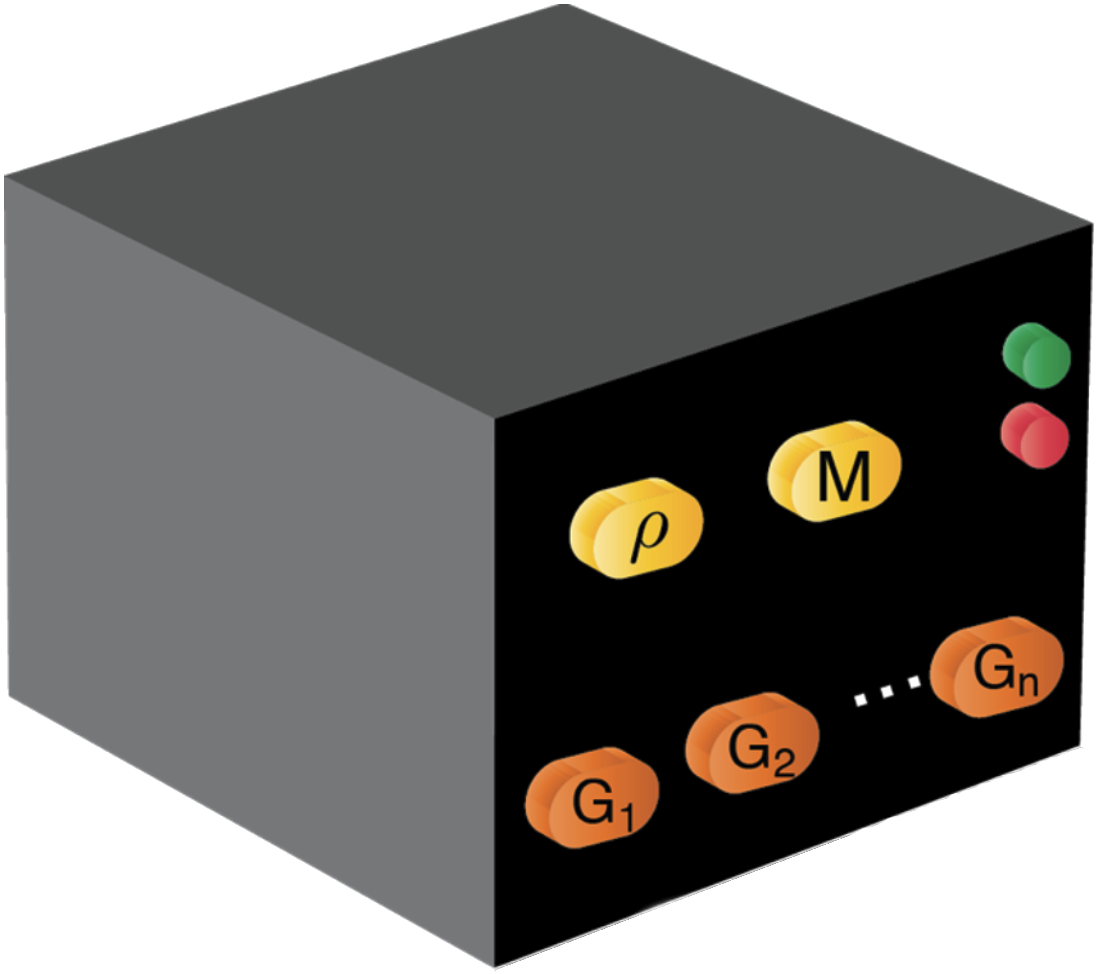}
    \caption{\textbf{We conceptualize a qubit (physical or logical) or a multi-qubit quantum register as a black box.} Logic operations such as initialization, gates, and measurements correspond to buttons that can be pushed in sequence to execute quantum circuits. Using this paradigm, we define \emph{button-theoretic Markovianity} to hold if and only if the quantum register experiences the same dynamical transformation every time a particular button is pressed. Protocols for characterizing physical qubits, such as gate set tomography, typically assume button-theoretic Markovianity. We say that an error-corrected register satisfies \emph{logical} Markovianity if button-theoretic Markovianity holds for the set of all \emph{logical} operations, including the logical idle.}
    \label{fig:blackbox}
\end{figure}

There is an extensive literature \cite{li_concepts_2018,milz_quantum_2021, Rivas2014-review, Wolf2008-hx,Wolf2008-mi,White2022-NMQPT,gorini1976completely, zwanzig1964identity, breuer2009measure, Breuer2016-Colloquium, pollock_non-markovian_2018, Pollock2018-operational} on non-Markovianity (and Markovianity) in quantum systems.  But in that work, Markovianity is almost always considered as a property of a system's autonomous dynamical evolution.  We need a notion of Markovianity that applies to \emph{programmable} quantum computers, whose dynamics depend very much on user inputs. For this purpose, we introduce a “button-theoretic” concept and definition of Markovianity based on the model used in gate set tomography \cite{blume2017demonstration}.  It stems from representing a quantum computer (e.g., a qubit or quantum register) as a black box equipped with buttons that can be pushed to apply assorted logic operations (state preparations, gates, and measurements) to the quantum register inside the box. A device satisfies button-theoretic Markovianity if and only if, every time a given button is pressed, the state of the quantum register is transformed exactly the same way.  In other words, every application of ``gate $a$'' can be represented by a unique, fixed CPTP map $G_a$, regardless of what other operations occur before or after.  It follows that the dynamical map produced by any \emph{sequence} of operations (e.g. ``$a$ then $b$ then $c$'') should coincide with the composition of those operations' maps (e.g., $G_cG_bG_a$). This property depends on the assumed dimension of the quantum computer, as any non-Markovian process can be described as the reduced dynamics of a Markovian process on a larger space (see Definition \ref{def:Buttons} in Appendix \ref{appendix:LM}, for a formal definition of button-theoretic Markovianity).

Button-theoretic Markovianity almost never holds exactly, but it is often a good approximation for the behavior of physical qubit systems. Moreover, it is \emph{assumed} by most QCVV protocols, such as standard implementations of GST and RB. These protocols infer gate properties from experiments in which long sequences of gates are performed, and thus rely upon the composition rule stated above. If a gate's action might be different every time it is applied, and/or might depend completely on past events, then the metrics or models inferred by GST or RB are unreliable, and may not correctly predict the gate's use in novel circuits. 

We wish to apply this notion of Markovianity to logical qubits.  To do so, we must distinguish physical processes from logical processes, and borrow some terminology and notation from the theory of fault tolerance \cite{gottesman2024surviving}. Table \ref{tab:notation} summarizes the notation of the mathematical objects used in this paper.

 \begin{table*}[]
     \centering
     \begin{tabular}{|c c c|}
        \hline
        Name & Notation & Superset/Domain/Codomain\\
        \hline
        Data qubit space & $\mathcal{H}_{data}$ &\\
        Logical qubit space & $\mathcal{H}_L$ &\\
        Syndrome qubit space & $\mathcal{H}_{Syn}$&\\
        Codespace & $\mathcal{C}$ & $\mathcal{H}_{data}$ \\
        $n$-qubit Pauli group (element)& $\mathcal{P}_n$ ($P_i$) & $\mathcal{H}_{data}$\\
        Stabilizer group (element)& $\mathcal{S}$ ($S$) &$\mathcal{P}_n$\\
        Stabilizer generators & $g_\alpha$&$\mathcal{S}$\\
        Encoding isometry & $E$&$\mathcal{H}_{L} \rightarrow \mathcal{H}_{data}$\\
        Encoding operation & $\mathcal{E}(\cdot) = E(\cdot)E^{\dagger}$ & $\mathcal{B}(\mathcal{H}_{L}) \rightarrow \mathcal{B}(\mathcal{H}_{data})$ \\
        Encoding unitary &  $U_E$&$\mathcal{H}_{Syn}\otimes\mathcal{H}_{L} \rightarrow \mathcal{H}_{data}$\\
        Encoding unitary superoperator & $\mathcal{U}_E = U_E (\cdot) U_E^{\dagger}$& $\mathcal{B}(\mathcal{H}_{Syn}\otimes\mathcal{H}_{L}) \rightarrow \mathcal{B}(\mathcal{H}_{data})$\\
        Correction Operators & $\{R(\s)\}$& $\mathcal{H}_{data}$\\
        Syndrome subspace projector & $\Pi_\s$ & $\mathcal{H}_{data}$\\
        Recovery map &  $\mathcal{R}(\rho) = \sum_\s R(\s) \Pi_\s \rho \Pi_\s R^{\dagger}(\s) $& $\mathcal{B}(\mathcal{H}_{data}) \rightarrow \mathcal{B}(\mathcal{C})\subseteq \mathcal{B}(\mathcal{H}_{data})$\\
        Decoding operation & $\mathcal{D} = \mathcal{E}^{\dagger}\circ\mathcal{R}$&  $\mathcal{B}(\mathcal{H}_{data}) \rightarrow \mathcal{B}(\mathcal{H}_{L})$\\
        Gate set defined on data qubits & $\mathcal{G}$& $\mathcal{B}(\mathcal{H}_{data})$\\
        Gate set defined on logical qubits & $\mathcal{G}_L$& $\mathcal{B}(\mathcal{H}_{L})$\\
        Gadget retraction & $\Omega[\mathcal{A}] = \mathcal{D} \circ \mathcal{A} \circ \mathcal{E}.$&$\mathcal{B}(\mathcal{B}(\mathcal{H}_{data})) \rightarrow \mathcal{B}(\mathcal{B}(\mathcal{H}_{L}))$\\
         \hline
\end{tabular}
     \caption{Notation of the various mathematical objects used in the paper, and their respective supersets, domains and codomains when applicable.}
     \label{tab:notation}
 \end{table*}

Suppose an experimenter wishes to implement a particular ``target'' quantum circuit (sequence of operations).  This can be done in multiple ways.  One of the most interesting is to implement it \emph{fault tolerantly}.  The experimenter implements a (larger) circuit that is equivalent to the target circuit if no errors occur, but which (unlike the original circuit itself) simulates the error-free circuit even if some errors occur. To make this more precise, we consider a target circuit on $k$ qubits that involves a state preparation $\rho_0$, a gate $L$ and a measurement $M$. These operations can be represented mathematically using $\mathcal{B}(\mathcal{H}^k)$, the space of bounded linear operators on a $k$-qubit Hilbert space, where $\mathcal{H} = \mathbb{C}^2$ is the Hilbert space of a single qubit. We represent $\rho_{0}$ by a density matrix in $\mathcal{B}(\mathcal{H}^k)$ describing the desired initial state, $L$ by a CPTP map in $\mathcal{B}(\mathcal{B}(\mathcal{H}^k))$ describing the gate's ideal action, and $M$ by a positive operator-valued measure (POVM) $M= \{ M_i : i = 1, \dots , n_{m}\}$, where $n_{m}$ is a positive integer, and each $M_{i}$ is a positive semidefinite operator in $\mathcal{B}(\mathcal{H}^k)$ associated with observing the outcome ``$i$''. The entire circuit can be represented visually:

\begin{equation}
\begin{quantikz}[column sep= .75cm]
    \inputD{\rho_{0}}&\gate{L}\qwbundle{k}&\meterD{M}
\end{quantikz}.
\end{equation}
 To protect this circuit from noise, we pick an $[[n,k]]$ code associated with an encoding operation, $\mathcal{E}: \mathcal{B}(\mathcal{H}^k) \rightarrow \mathcal{B}(\mathcal{H}^n)$, and a decoding operation $\mathcal{D}: \mathcal{B}(\mathcal{H}^n) \rightarrow \mathcal{B}(\mathcal{H}^k)$. We then design a new circuit over $n$ qubits by replacing each $k$-qubit operation with an $n$-qubit ``gadget'' circuit.  $\rho_{0}$ is replaced by a state preparation gadget that produces the state $\Bar{\rho_{0}} := \mathcal{E}(\rho_{0})\in  \mathcal{B}(\mathcal{H}^n)$. $M$ is replaced by a measurement gadget that implements the POVM $\Bar{M}:= \{ \mathcal{D}^{\dagger}(M_{i}) : i = 1, \dots , n_{m}\} $. And the the gate $L$ is replaced by a gate gadget that applies the CPTP map $G_L\in \mathcal{B}(\mathcal{B}(\mathcal{H}^n))$, where

\begin{equation}\label{eq: gadget retraction}
    \Omega[G_L]:= \mathcal{D} \circ G_L \circ \mathcal{E}=L.
\end{equation}

The map, $\Omega: \mathcal{B}(\mathcal{B}(\mathcal{H}^n)) \rightarrow \mathcal{B}(\mathcal{B}(\mathcal{H}^k))$, was also used in \cite{rahn_exact_2002,beale_efficiently_2021} to define effective logical processes, where for an arbitrary $n$-qubit operation $G$, $\Omega[G]$ is the effective $k$-qubit logical process for $G$. Henceforth, we will call $\Omega$ the \emph{gadget retraction}. In the absence of error, the encoded circuit produces measurement outcomes with the same distribution as those of the original unencoded circuit (we call any two circuits with this property \textit{equivalent}). Unfortunately, in real-world scenarios, no operation can be performed perfectly.  If errors occur in the original unencoded circuit, its outcomes will change.  The point of the fault-tolerant encoding is that, if we do it effectively, at least \textit{some} errors occurring in the gadgets that make up the fault-tolerant circuit will \textit{not} change the circuit's outcome.  

In this paper, we focus on the behavior of faulty \textit{gate} gadgets, so to keep matters simple we will assume that both the state preparation and measurement gadgets are error-free. We will also assume that the $n$-qubit operations on the physical qubits satisfy button-theoretic Markovianity.  Under these assumptions, attempting to implement $G_L$ on the quantum register using noisy error-prone gates will actually apply a noisy gadget for $L$, $\widetilde{G}_L\in \mathcal{B}(\mathcal{B}(\mathcal{H}^n))$:
\begin{equation}
\begin{quantikz}[column sep= .75cm]
    \inputD{\Bar{\rho_{0}}}&\gate{\widetilde{G}_{L}}\qwbundle{n}&\meterD{\Bar{M}}
\end{quantikz}.
\end{equation}

This circuit is equivalent to a $k$-qubit circuit
\begin{equation}
\begin{quantikz}[column sep= .75cm]
    \inputD{\rho_{0}}&\gate{\Omega[\widetilde{G}_{L}]}\qwbundle{k}&\meterD{M}
\end{quantikz}.
\end{equation}
The map $\Omega[\widetilde{G}_{L}]$ is called the effective logical process for $\widetilde{G}_{L}$.

Now consider the process that results from pressing button ``gate $a$" then ``gate $b$". Using the same assumptions as before, it can be represented by the following circuit:
\begin{equation}
\begin{quantikz}[column sep= .75cm]
    \inputD{\Bar{\rho_{0}}}&\gate{\widetilde{G}_{a}}\qwbundle{n}&\gate{\widetilde{G}_{b}}&\meterD{\Bar{M}}
\end{quantikz},
\end{equation}
which is equivalent to the $k$-qubit circuit,
\begin{equation}
\begin{quantikz}[column sep= .75cm]
    \inputD{\rho_{0}}&\gate{\Omega[\widetilde{G_{b}} \circ \widetilde{G_{a}}]}\qwbundle{k}&\meterD{M}
\end{quantikz}.
\end{equation}
The gates are only Markovian over the logical qubit space if this circuit is equivalent to the circuit
\begin{equation}
\begin{quantikz}[column sep= .75cm]
    \inputD{\rho_{0}}&\gate{\Omega[\widetilde{G_{a}}]}\qwbundle{k} &\gate{\Omega[ \widetilde{G_{b}}]}&\meterD{M}
\end{quantikz},
\end{equation}
which is to say if $\Omega[\widetilde{G_b} \circ \widetilde{G_a}]= \Omega[\widetilde{G_b}] \circ \Omega[\widetilde{G_a}]$. We say that a gate or set of gates is \emph{logical gate composable} (see Definition \ref{def:gate composability}) if this condition holds for every gate or pair of gates in the set. Therefore, logical gate composability is the coincidence of button-theoretic Markovianity at the physical and logical levels, using the very same button labels. It's violation indicates not just non-Markovianity, but \emph{emergent} non-Markovianity. We will demonstrate that gadgets -- in particular, QEC gadgets that implement a logical idle operation -- built from physically Markovian gates in realistic settings can violate gate composability, producing effective logical processes that are non-Markovian.

 Finally, we note that in this paper we do not consider cases where $\widetilde{G}_{L}$ and $\Omega[\widetilde{G}_{L}]$ depend on measurement outcomes. This naturally captures measurement-free fault-tolerance schemes \cite{heussen2024measurement, butt_demonstration_2026}, but in many QEC and fault-tolerance schemes syndrome measurements are kept track of and used for decoding, creating processes that are conditioned on syndrome histories which (depending on the protocol) can span several gadgets. In these cases, we assume a coarse-grained picture where the effective logical process results from marginalizing over any syndrome measurement outcomes. This is primarily for practical purposes: keeping track of logical processes that are syndrome-dependent would in general require information that is exponential in the number of syndrome bits, defeating the practical purpose of modeling the process on a smaller (logical) space (though this can be reduced somewhat if the noise or code has significant symmetry \cite{beale_efficiently_2021}). We consider studying logical process maps which explicitly track syndrome histories to be a promising direction for future study.

\subsection{Non-Markovianity from syndrome errors}

\begin{figure*}
    \centering
    \includegraphics[width=0.75\linewidth]{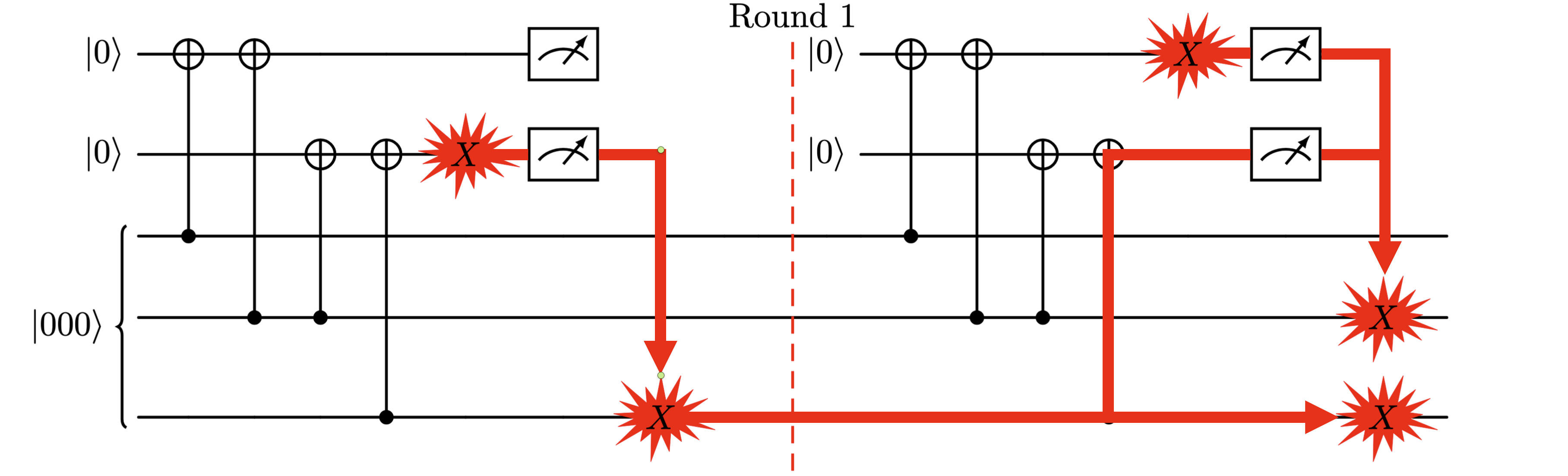}
    \caption{Two rounds of syndrome extraction in the 3-qubit repetition code with bitflip errors on the ancilla qubits. The bitflip in the first round leads to an incorrect readout of the syndrome, causing the decoder pick the wrong correction. The result is a single qubit error on the data qubits that won't change the outcome of an error-corrected logical measurement. In the second round, an ancilla error causes a second bitflip on the data qubits, leading to a combined error that will flip the outcome of an error-corrected logical Z measurement.}
    \label{fig:2rnd}
\end{figure*} 
To see how non-Markovian logical operations can emerge from Markovian physical-qubit operations, we start with a scenario in which they \textit{don't} occur.  Gate composability always holds if each gate in the gate set is a noisy gadget followed by a perfect (error-free) QEC gadget that leaves the system in the code space.  We prove this in Appendix \ref{sec: sufcond} (Theorem \ref{thm: sufcond}), but the intuition is straightforward: if an $n$-qubit process both starts and ends in the codespace (which is isomorphic to the factor space associated with the logical $k$-qubit subsystem), then that process is equivalent to a process acting only on the logical qubit[s]. It follows that logical non-Markovianity can emerge from Markovian physical gates only if the QEC itself is imperfect -- i.e., if there are errors in the measured syndrome.

Emergent logical non-Markovianity can be demonstrated explicitly in a simple scenario involving two rounds of a noisy error correction gadget for a CSS code \cite{calderbank1996good,steane1996multiple}. The error correction gadget comprises (1) syndrome extraction, and (2) correction.  In the syndrome extraction step, ancilla qubits are used to perform parity measurements that extract information about errors.  In the correction step, a conditional operation is applied to the data qubits, which depends on the measured syndromes.  If we combine these steps and marginalize over the syndrome measurement's outcomes, we get a \textit{recovery map} acting on the $n$ data qubits, which can be seen as a fault-tolerant implementation of a logical idle gate (or equivalently, a gadget for the identity gate). 

We denote the perfect (noise-free) recovery map by $\mathcal{R}$ (see Equation \ref{eq: idealrecov}).  We consider an error model in which the only errors are in syndrome readout, which isolates the role that syndrome errors play in emergent non-Markovianity. Using the 3-qubit repetition code as an example, the noisy recovery map $\widetilde{\mathcal{R}}$ is given by the following circuit:

\begin{center}
\begin{tikzpicture}
\node[scale=.7] {
\begin{quantikz}
&\gate[3]{\widetilde{\mathcal{R}}}&\\
&&\\
&&\\
\end{quantikz} =
\begin{quantikz}
\lstick{0}&\targ{}&\targ{}&&&\gate[style={fill=red}]{\mathcal{F}}&\meter{}&\setwiretype{c}\rstick[2]{$\mathbf{\s}$}\\
\lstick{0}&&&\targ{}&\targ{}&\gate[style={fill=red}]{\mathcal{F}}&\meter{}&\setwiretype{c}\\
&\ctrl{-2}&&&&&&&\gate[3]{R(\mathbf{\s})}&\\
&&\ctrl{-3}&\ctrl{-2}&&&&&&\\
&&&&\ctrl{-3}&&&&&\\
\end{quantikz}};
\end{tikzpicture}
\end{center}

\noindent where $\mathcal{F}(\rho)= (1-p)\rho + p X\rho X$ is the single qubit bitflip channel with probability $p$, and the corrections are
 \begin{equation}\label{eq: rep corrections}
         R(\mathbf{\s}) =
    \begin{cases}
    I\otimes I\otimes I & \text{if } \mathbf{\s} = (0,0)\\
     X\otimes I\otimes I& \text{if } \mathbf{\s} = (1,0)\\
     I\otimes I\otimes X & \text{if } \mathbf{\s} = (0,1)\\
    I\otimes X\otimes I& \text{if } \mathbf{\s} = (1,1)\\
    \end{cases}.
 \end{equation}
The circuit also represents the noiseless recovery map, $\mathcal{R}$, when $p=0$. When $p>0$, just two applications of the noisy error correction gadget are sufficient to violate gate composability.

In the first round, the only possible errors on the data qubits are corrections (applied incorrectly in response to faulty syndrome data). These do not affect the outcome of any logical measurement, because they only cause \textit{correctable} errors and we assume that each logical measurement incorporates a perfect round of error correction. Therefore, the effective logical process for \textit{one} round of noisy QEC is a perfect logical identity operation (no errors occur). 

Now, if gate composability holds, then the logical process for two rounds should also be the identity operation.  But this is not the case. If the first round induces a correctable error, it leaves the data qubits in a state outside the code space.  This does not constitute a logical error, because if a perfect logical measurement was performed immediately at this point, the results would be correctly decoded.  However, if a \textit{second} correctable error occurs in the next round, it can combine with the leftover error from the first round to create an uncorrectable error.  At this point, performing a perfect Z logical measurement and decoding will yield the incorrect outcome (an example of this is shown in Fig. \ref{fig:2rnd}). This means the logical process for two rounds will \textit{not} be a perfect identity, which implies the violation of gate composability. 

Now that we have demonstrated a clear example of logical non-Markovianity, a few questions remain.  Can these logical non-Markovian dynamics be well-approximated by a Markovian process? If we repeat the gate, is there a point in time after which the logical process becomes approximately Markovian?  Unfortunately, we lack a metric for logical non-Markovianity that can be used to quantify how close to Markovian each timestep is. Therefore, in lieu of a metric and a general theory of logical non-Markovianity, we can begin to answer these questions by generalizing the example given above to other stabilizer codes, and by examining how the expectation value of the logical Z operator decays over many rounds of noisy syndrome extraction.

  \section{Logical non-Markovianity in repeated syndrome extraction circuits}\label{sec: example}

We consider QEC experiments on a single ($k=1$) logical qubit that comprise (1) initialization in a computational basis state, (2) repeated applications of noisy syndrome extraction, and (3) measurement of the logical $Z$ operator.  If the logical errors are Pauli and Markovian, then the expectation value of the logical Z operator (``polarization'') will decay exponentially.  Therefore, any deviation from exponential decay indicates non-Markovianity. We consider both the 3-qubit repetition $([3,1,3])$ code \cite{peres1985reversible} and the $[[5,1,3]]$ perfect code \cite{laflamme1996perfect}. Both are stabilizer codes \cite{gottesman1997stabilizer}. We will comment on emergent logical Markovianity in general ($k>1$) stabilizer codes in section \ref{sec: LM stab}.

For this analysis, we assume that the following logical operations can be implemented:
\begin{itemize}
    \item State preparation: we assume that we can perfectly prepare $\ket{\Bar{0}}$ and $\ket{\Bar{1}}$.
    \item Logical idle: We can perform noisy rounds of syndrome extraction representable by the map $\widetilde{\mathcal{R}}$.
    \item Measurement: We assume that a terminal logical Z measurement can be performed without error. We also assume that each logical measurement is preceded by an ideal round of QEC, representable by the map $\mathcal{R}$. 
\end{itemize}

For the codes considered, we construct mathematical models of the various components of the circuits ($\ket{\Bar{0}}$, $\mathcal{R}$, etc.) using the stabilizer code formalism, which we review quickly here.

\subsection{Stabilizer codes}
Stabilizer codes \cite{gottesman1997stabilizer} are quantum codes whose codespace is the $+1$ eigenstate of a set of Pauli operators called a stabilizer group.  A stabilizer group $\mathcal{S}$ is an abelian subgroup of the $n$-qubit Pauli group $\mathcal{P}_n = \{I,X,Y,Z\}^{\otimes n}$ that does \textit{not} contain $-I$ (the negative identity).  This group $\mathcal{S}$ defines a stabilizer subspace $\mathcal{C} \equiv \{\ket{\psi} : S\ket{\psi}= \ket{\psi} \quad \forall S \in \mathcal{S}\}$.  If the subspace $\mathcal{C}$ satisfies the properties of a quantum error correcting code, it is a \textit{stabilizer code}.  The stabilizer group of an $[[n,k]]$ code can be generated by a set of $n-k$ stabilizer generators, i.e., $\{g_\alpha, \alpha = 1, 2, \dots, n-k\}$, such that  $\langle \{ g_\alpha \} \rangle =\mathcal{S}$.

A code's stabilizer generators are commuting Pauli operators, so they can in principle be measured, yielding a string of $n-k$ bits.  This string is called a \emph{syndrome}, and the state of the $n$ data qubits is in the code space if and only the syndrome is $0^{\otimes n-k}$. Performing perfect syndrome extraction and observing the syndrome $\mathbf{s}= (s_1, \cdots, s_{n-k})$ projects the data qubits into a \textit{syndrome subspace} indexed by $\mathbf{s}$. The associated projector is
\begin{equation}\label{eq: syn proj}
\Pi_{s}= \prod_{j=1}^{n-k}\frac{(I + (-1)^{s_{j}}g_{j})}{2}
\end{equation}
We say an $n$-qubit Pauli operator $P$ \emph{has} the syndrome $\mathbf{s}$ if
\begin{equation}
    Pg_i = (-1)^{s_{i}}g_{i}P,\quad i = 1,\dots,n-k.
\end{equation}

To perform error correction on a stabilizer code, the syndrome is measured, and then a \textit{correction} (a.k.a. ``recovery operator''), chosen to correct the observed syndrome, is applied.  Its purpose is to return the data qubits to the code space.  For stabilizer codes, the correction can always be chosen to be a Pauli operator. The ideal recovery map for one error-free round of QEC (syndrome extraction followed by correction) on a stabilizer code can be written in the form
\begin{equation}\label{eq: idealrecov}
    \mathcal{R}(\rho) = \sum_{\mathbf{s}}R(\mathbf{s}) \Pi_{\mathbf{s}} \rho \Pi_{\mathbf{s}} R(\mathbf{s}),
\end{equation}
where $R(\mathbf{s})$ is a Pauli correction with syndrome $\mathbf{s}$. 

Pauli operators that map a stabilizer code space to itself (i.e., do not change the syndrome) correspond to the normalizer $N(\mathcal{S})$ of the stabilizer group, which is the set of $n$-qubit Pauli operators that commute with every operator in $\mathcal{S}$ and therefore have a trivial syndrome. The set of cosets, $N(\mathcal{S})/S$,  is isomorphic to the $k$-qubit Pauli group. For a stabilizer code encoding a single logical qubit, one Pauli operator in $N(\mathcal{S})$ can be chosen to represent the logical $Z$ operator, denoted $\Bar{Z}$. The choice of $\Bar{Z}$, together with the stabilizers of the code, defines the logical computational basis states, $\ket{\Bar{0}}$ and $\ket{\Bar{1}}$, which are the eigenstates in the codespace of the $+1$ and $-1$ eigenvalues of $\Bar{Z}$, respectively.

\subsection{Repeated noisy syndrome extraction}

In systems where only single qubit measurements are natively performed, syndrome extraction requires measurement of ancilla qubits. We model syndrome extraction errors by assuming that each ancilla is flipped, before measurement, with probability $p$.  If one or more ancillae experience errors, the syndrome will be measured incorrectly, and the wrong correction operator will be applied. If we marginalize over the results of the syndrome measurements, a single round of error correction applies the map
\begin{equation}\label{eq: noisyrecov}
    \widetilde{\mathcal{R}}(\rho) = \sum_{\s, \s'} \chi_{\s, \s'} R(\s') \Pi_{\s} \rho \Pi_{\s} R(\s'),
\end{equation}
where $\chi_{\s, \s'}= \mathrm{Pr}(s'|s) =  (1-p)^{2-h(\s,\s')} p^{h(\s,\s')}$ is a confusion matrix \cite{beale_efficiently_2021} that represents the probability of measuring $\s'$ when the true syndrome is $\s$, and $h(\s,\s')$ is the Hamming distance between $\s$ and $\s'$. In our analysis, we will use ``$\widetilde{\mathcal{R}}$'' to represent the quantum error correction gadget.

The circuits we consider are of the following form:
\begin{center}
\begin{tikzpicture}
\node[scale=.8] {
\begin{quantikz}
\gate[3]{|\Bar{0}\rangle \text{ or } |\Bar{1}\rangle}&\gate[3]{\widetilde{\mathcal{R}}}&\gate[3]{\widetilde{\mathcal{R}}}& \gate[3]{\widetilde{\mathcal{R}}}&\  \ldots \ 
&\gate[3]{\widetilde{\mathcal{R}}}&\gate[3]{\mathcal{R}}&\meter[3]{\Bar{Z}}\\
&&&&\  \ldots \ &&&\\
&&&&\  \ldots \ &&&\\
\end{quantikz}
};
\end{tikzpicture}
\end{center}

These circuits can be combined analytically to determine the formula for the polarization $q_m = \langle\Bar{Z}\rangle$ after $m$ rounds of noisy syndrome extraction:

\begin{equation}\label{eq:circuit probs}
   q_m =  Tr\left[\Bar{Z}\mathcal{R}\circ\widetilde{\mathcal{R}}^{m}(\mathcal{E}(Z/2)\right]. 
\end{equation}

The polarization can be thought of as the (error-corrected) expectation value of $\Bar{Z}$ after $m$ rounds of QEC when the system state is initialized in the (non-physical) state $\mathcal{E}(Z/2)$. The map $\mathcal{E}$ is the encoding operation for a code (see Appendix \ref{appendix:logical dynamics} for full details on the definition), and $\mathcal{E}(Z)$ is the restriction of $\Bar{Z}$ to the code space. For example, in the repetition code,
\begin{equation}
    \mathcal{E}(Z) = \dyad{000}  -\dyad{111}.
\end{equation}
The operator $\mathcal{E}(Z/2)$ is not a quantum state, but a circuit that starts with $\mathcal{E}(Z/2)$ as initial state can be emulated by running two circuits, one initialized with the state $\ket{\Bar{0}}$ and one with $\ket{\Bar{1}}$. A weighted average of the outcomes of both circuits yields the expectation value corresponding to the ``synthetic''\footnote{This is an example of a synthetic circuit \cite{fan2024randomized}, a linear combination of physical circuits that allows for the emulation of non-physical operators.} initial state $\mathcal{E}(Z/2)$. The factor of $1/2$ (using $\mathcal{E}(Z/2)$ instead of $\mathcal{E}(Z)$) ensures that $q_m =1$ in the absence of error. Preceding the logical measurement with a perfect round of QEC ensures that the data qubits' state is always in the codespace before the measurement, so that measuring $\mathcal{E}(Z)$ is equivalent to measuring $\Bar{Z}$. Therefore, an equivalent expression for the polarization can be derived:
\begin{align}
   q_m &= Tr\left[\mathcal{E}(Z)\mathcal{R}\circ\widetilde{\mathcal{R}}^{m}(\mathcal{E}(Z/2)\right]\\
    & =  Tr\left[\mathcal{D}^{\dagger}(Z)\widetilde{\mathcal{R}}^{m}(\mathcal{E}(Z/2)\right].\label{eq: polarization}
\end{align}

We call the map $\mathcal{D}\equiv \mathcal{E}^{\dagger}\circ \mathcal{R}$ the \textit{decoding operation}. The operator $\mathcal{D}^{\dagger}(Z)$ can be thought of, informally, as an ``error-corrected'' $\Bar{Z}$.  For the 3-qubit repetition code,

\begin{equation}
    \begin{split}
    \mathcal{D}^{\dagger}(Z) &= \dyad{000}  -\dyad{111}\\
                            & +\dyad{001} -\dyad{110}\\
                            & +\dyad{010}  -\dyad{101}\\
                            & +\dyad{100}  -\dyad{011}.
\end{split} 
\end{equation}

Any state that decodes to $\ket{0}$ has an eigenvalue of $+1$ and any state that decodes to $\ket{1}$ has an eigenvalue of $-1$. For CSS codes like the 3-qubit repetition code, instead of measuring $\Bar{Z}$ with multiqubit operations, a transversal measurement can be performed, where each data qubit is measured in the $Z$ basis. This is followed by classical error correction to determine the value of $\Bar{Z}$. If the corrections used to define $\mathcal{D}$ are the same as the ones involved in this classical error-correction step, then this transversal measurement is equivalent to measuring $\mathcal{D}^{\dagger}(Z)$. 

Focusing on $\mathcal{D}^{\dagger}(Z)$ enables relating $q$ to effective logical processes. Using the definition of the gadget retraction (Eq. \ref{eq: gadget retraction}),  Equation \ref{eq: polarization} can be rewritten:
\begin{align}
  q_m &=  Tr\left[\mathcal{D}^{\dagger}(Z)\widetilde{\mathcal{R}}^{m}(\mathcal{E}(Z/2))\right] \\
  & = Tr\left[Z \Omega[\widetilde{\mathcal{R}}^{m}]( Z/2)\right].
\end{align}
By determining $q_m$, we are essentially characterizing the dynamical maps, $\Omega[\widetilde{\mathcal{R}}^{m}]$, for $m=1,2\ldots$, which are the effective ``logical idle'' processes describing $m$ repeated rounds of noisy QEC. As discussed previously, $\widetilde{\mathcal{R}}$ is not necessarily gate composable, so $\Omega[\widetilde{\mathcal{R}}^{m}]\neq \Omega[\widetilde{\mathcal{R}}]^{m}$, which allows for non-Markovian behavior. 

We now analyze this phenomenon in two specific cases, the three-qubit repetition code, and the 5-qubit $[[5,1,3]]$ code. Code implementations for both examples are available on the PRAQTICE github repository \cite{ostrove_praqtice_2025}.

\begin{figure*}
    \centering
     \begin{subfigure}[t]{0.334\linewidth}
        \centering
        \includegraphics[width=\linewidth]{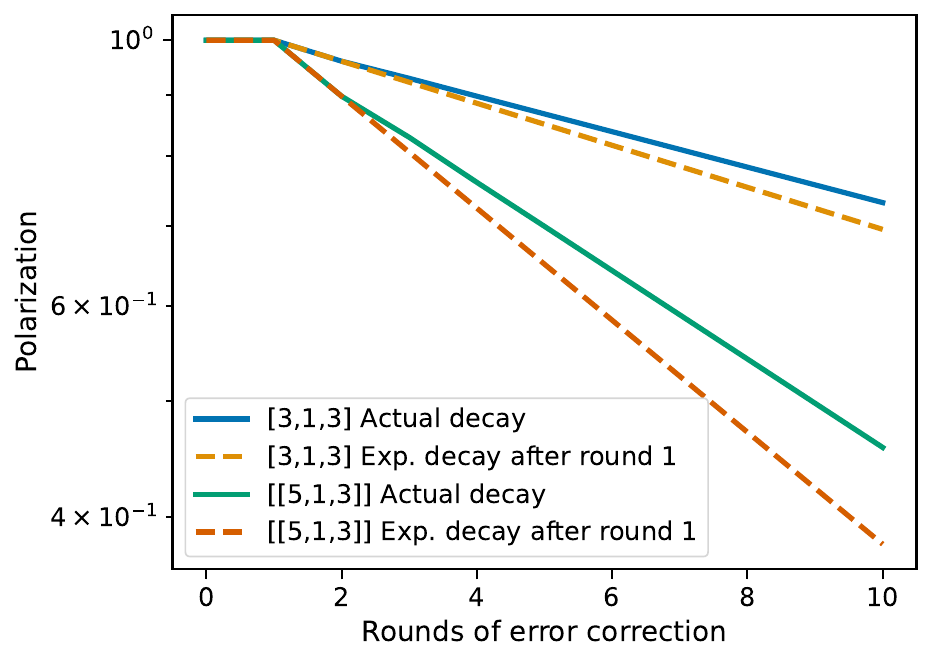}
        \caption{}
        \label{fig:surv}
    \end{subfigure}
    \hfill
    \begin{subfigure}[t]{0.32\linewidth}
        \centering
        \includegraphics[width=\linewidth]{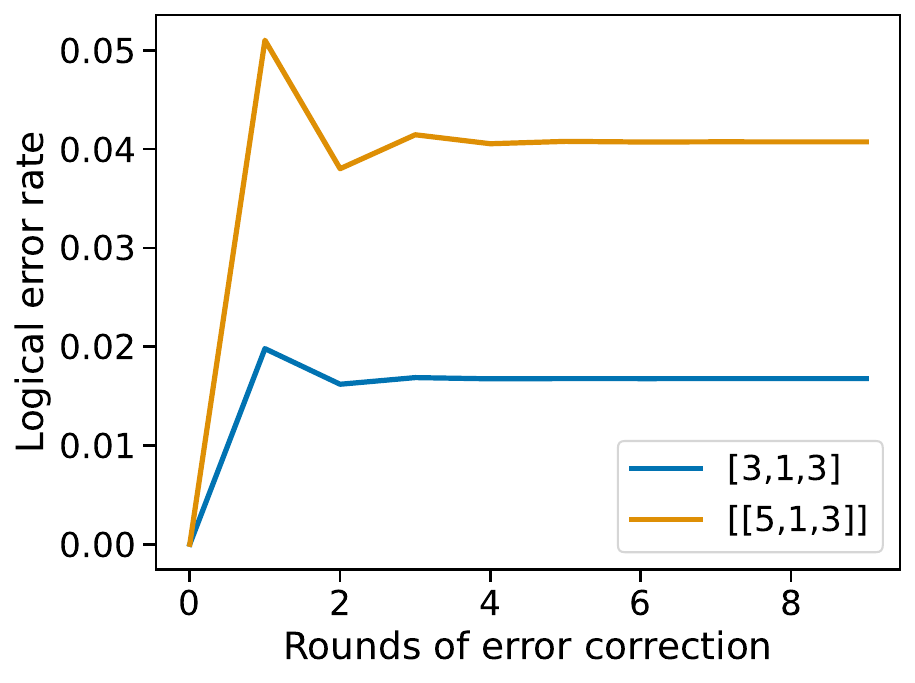}
        \caption{}
        \label{fig:decay}
    \end{subfigure}
    \hfill
    %\vspace{1cm}
    \begin{subfigure}[t]{0.334\linewidth}
        \centering
        \includegraphics[width=\linewidth]{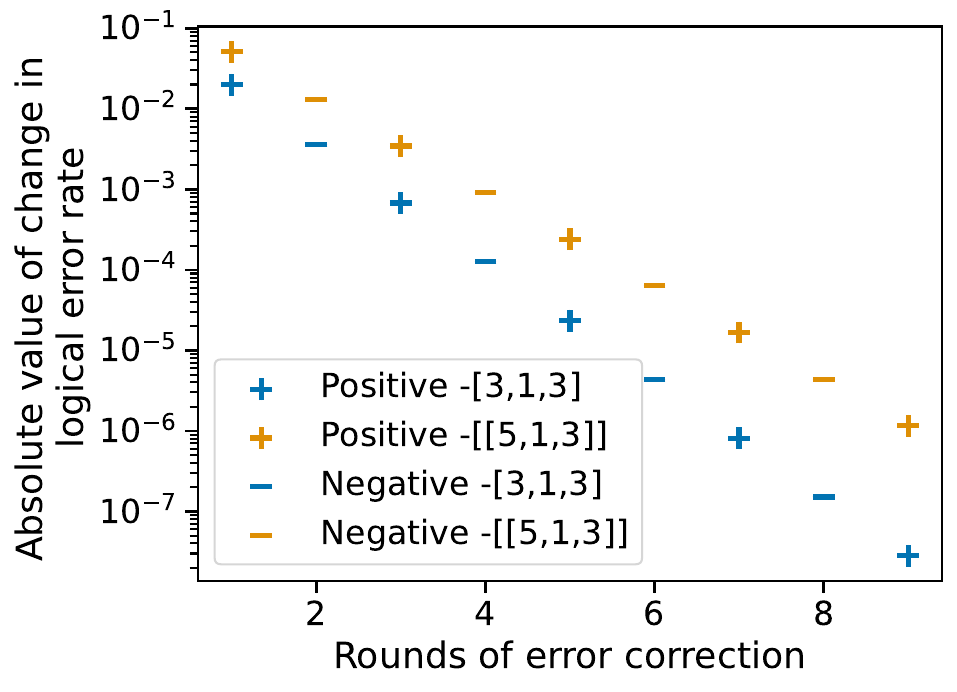}
        \caption{}
        \label{fig:decaychange}
    \end{subfigure}

    \caption{ (\subref{fig:surv}) Polarization in the $[3,1,3]$ code and the $[[5,1,3]]$ code under repeated rounds of noisy syndrome extraction. A constant exponential decay is expected for a Markovian logical process. In log scale, non-Markovianity is demonstrated by a deviation of the decay from a straight line. The decay curve becomes quickly indistinguishable from straight after just a few rounds. This indicates logical processes where most of the non-Markovianity is relegated to the first few rounds. (\subref{fig:decay}) The logical error rate per round of the $[3,1,3]$ code and the $[[5,1,3]]$ code. The absolute value of the change in error rate is plotted in (\subref{fig:decaychange}). It is clear from examining the raw data that the error rate change falls of exponentially in the number of rounds until the changes fall below the machine precision ($\sim 10^{-15}$). In summary, the data indicate a logical process that is strictly non-Markovian for any number of rounds, but where most of the deviation from a Markovian model occurs within the earliest rounds.\label{fig:polarization}}
\end{figure*}

\subsubsection{The 3-qubit repetition code}

The decay of the logical qubit polarization resulting from repeated rounds of noisy syndrome extraction in the 3-qubit repetition code is shown in Fig.~\ref{fig:polarization}. Also shown is the logical error rate in each round, $m$, which is equal to $(1-q_{m+1}/q_{m})/2$. In this model, only single-qubit errors are injected on the data qubits in each round, so the first round will have no logical errors, after which the polarization will decay. An exponential decay with a constant rate is expected for a Markovian gate. The change in logical error rate between round $m$ and $m+1$, shown in Fig. \ref{fig:decaychange}) demonstrates a logical error rate that changes each round but whose changes die off exponentially. This indicates a logical process that is strictly non-Markovian for any number of rounds, but where most of the deviation from a Markovian model occurs within the earliest rounds. If this is true of the logical processes of more realistic quantum computations, it would allow for logical tomography because the first few rounds could be characterized separately from later rounds (i.e., with those rounds acting differently than subsequent rounds). Just considering a single code, it is unclear how this generalizes to other codes, but it can be proven that the same qualitative behavior is expected for any Pauli channel in a realistic quantum error correction setting \cite{kwiatkowski2025constructing}. Next we provide an example of a distance-3 \textit{quantum} code that exhibits the same non-Markovianity.

\subsubsection{The [[5,1,3]] code}

Now we consider the decay of the polarization in the $[[5,1,3]]$ code \cite{yoder2016universal}. First, we implicitly define its noisy recovery map, $\widetilde{\mathcal{R}}$, by defining its stabilizer group, logical Z, and corrections. The stabilizer group and logical Z for the $[[5,1,3]]$ code are

\begin{equation}\label{eq:stabgens}
    \mathcal{S}_5 = \Biggl \langle \begin{matrix}
        ZZXIX\\
        XZZXI\\
        IXZZX\\
        XIXZZ
    \end{matrix} \Biggr \rangle ,
\end{equation}
\begin{equation}
    \Bar{Z}_5 = -XIZIX.
\end{equation}

The set of corrections can be inferred by recognizing that each syndrome corresponds to a unique single qubit Pauli that has that particular syndrome. The decay of the polarization that results from repeated noisy syndrome extraction rounds is given in Fig.~\ref{fig:polarization}.

The decay is qualitatively similar to what we observed for the $[3,1,3]$ code: there is a large amount of initial non-Markovianity that decays away rapidly to yield an approximately Markovian process at long times. However, in the 5-qubit code there are more ways for single-qubit errors to combine into logical errors, and this increases the initial change in error rate compared to the $[3,1,3]$ code.  Just as in the 3-qubit repetition code, each syndrome error induces at most a weight-one error on the data qubits, and every weight-two error leads to decoder failure. This is because every non-trivial element of the stabilizer group is a weight-4 Pauli, and each correctable error is stabilizer-equivalent to a weight-one operator. When a weight-one operator is multiplied by another Pauli, the weight of the product can differ from the latter by at most one, so no weight-two operators are correctable. A weight-two error occurs in two rounds only if the syndrome errors in the two rounds are both non-trivial and non-identical. In the $[3,1,3]$ code the probability of a error in the first round is, to leading order, the probability that either of the 2 ancillae has an error, i.e. $2p$. The probability that a \emph{different} error occurs in the next round is, to leading order, $p$. In the 5-qubit code there are 4 ancilla qubits, so, to leading order the probability that anyone of them will have a single qubit error will be $4p$. The probability that a different error will occur in the second round is, to leading order, $3p$. It turns out that, for the chosen stabilizer generators (Eq. \ref{eq:stabgens}), half of the weight-two errors generated by weight-one syndrome errors are X and Y logical errors, which will affect the polarization. So to leading order, logical errors are more likely to occur in the $[[5,1,3]]$ code than in the $[3,1,3]$ repetition code. It should be noted that this is fully the consequence of these circuits being non-fault-tolerant, so we expect a fault-tolerant treatment of these experiments to behave differently.

\subsubsection{Logical non-Markovianity in stabilizer codes}\label{sec: LM stab}

We have demonstrated logical non-Markovianity in at least two codes. We now investigate whether it is expected to occur in a broader class of codes. As we will show, in any stabilizer code that can correct single qubit errors, as long as there is the possibility of any syndrome extraction error, then for a model of a QEC round whose only noise derives from these syndrome errors, repeated application of the noisy QEC round will produce logical non-Markovianity. The intuition behind this is that for a code with an adequate distance and set of corrections, there are always two corrections that can combine to form an uncorrectable error. Because syndrome extraction is noisy, there will be a possibility that the decoder will erroneously apply these two corrections in subsequent rounds. The application of the correction in the first round will only affect the syndrome information, but the application of the second correction in round two can affect both syndrome and logical information, in a scenario analogous to the repetition code example described in Fig \ref{fig:2rnd}.  This will violate logical gate composability, because in order to accurately reconstruct the process on the logical qubits, the correlations between the logical and syndrome information of a state must be known. This is summarized by the following theorem:

\begin{theorem}\label{thm: stab}
    Consider an $[[n,k,d]]$ stabilizer code. Let $\{g_\alpha, \alpha = 1, 2, \dots, n-k\}$ be a set of $n-k$ stabilizer generators for this code. Consider a set of corrections, $\mathfrak{Corrections}= \{R(\s)  : 
\s \in \{0,1\}^{\otimes n-k}\} $, such that for $R(s)\in \mathfrak{Corrections}$, $ R(s) \in  \mathcal{P}_n$, and $R(s)g_i = (-1)^{s_{i}}g_{i}R(s)$ for all $i$. Consider a noisy round of syndrome extraction represented by the map,
\begin{equation}\label{eq: general noisyrecov}
    \widetilde{\mathcal{R}}(\rho) = \sum_{\s, \s'} C_{\s, \s'} R(\s') \Pi_{\s} \rho \Pi_{\s} R(\s')
\end{equation}
where $\sum_\s C_{\s, \s'} = 1$, $C_{\s, \s'}>0$ $\forall \s, \s' \in \mathbb{Z}_{2}^{\otimes n-k}$, and $\Pi_{\s}$ is defined as in Equation \ref{eq: syn proj}. Let $\Omega$ be the gadget retraction for the code and set of corrections. If $d\geq 3$, then
\begin{equation}\label{eq: composability violation}
    \Omega[\widetilde{\mathcal{R}} \circ \widetilde{\mathcal{R}}]\neq \Omega[\widetilde{\mathcal{R}}] \circ \Omega[\widetilde{\mathcal{R}}].
\end{equation}
Logical gate composability is violated for the process.
\end{theorem}

Proving this statement requires two observations. The first is that choosing a stabilizer code and a set of corrections defines a splitting of each Pauli into a logical component and a correction component:

\begin{proposition}\label{prop}
    Consider an $[[n,k,d]]$ stabilizer code with a stabilizer group $\mathcal{S}$ and a set of corrections $\{R(\s) : 
\s \in \{0,1\}^{\otimes n-k}\} $. Then $\forall P \in \mathcal{P}_n$,
\begin{equation}\label{eq:logical decomp}
    P = L R(\s)
\end{equation}
where $L$ commutes with the stabilizer group ($L\in N(\mathcal{S})$) and $R(\s)$ is a correction. We call $L$ the logical component of $P$, and $R(\s)$ the correction component of $P$. 
\end{proposition}
The correction component is the factor of a Pauli that would be corrected during (and therefore not affect) an error-corrected logical measurement when it is applied to a codespace state. The logical component, on the other hand, can affect the logical measurement outcome. 
\begin{proof}

The proposition is a straightforward consequence of a fact that is widely used in the stabilizer formalism:  given a a stabilizer group $\mathcal{S}\subset \mathcal{P}_n$, every Pauli $P \in \mathcal{P}_n$ can be written as the product of three Pauli operators:
\begin{equation}\label{eq:LSD}
    P = L'SD(\s)
\end{equation}
where $L'$ is a logical operator, $S$ is an element of the stabilizer group, and $D(\s)$ is an element of the destabilizer group with some syndrome $\s$. Given such a Pauli, a correction operator with the same syndrome, $R(\s)$, can be decomposed similarly:

\begin{equation}\label{eq:RLSD}
    R(\s) = L''S'D(\s).
\end{equation}
Using equations \ref{eq:LSD} and \ref{eq:RLSD}, and the self-adjointness of the operators,
\begin{equation}
    P = L'SD(\s) = L'L''SS'R(\s).
\end{equation}
Setting $L = L'L''SS'$, clearly Equation \ref{eq:logical decomp} is true for $P$.
\end{proof}
The other component necessary for the proof is the fact that in the codes and corrections considered, there will always be two corrections that can combine into an uncorrectable error. This is summarized by the following lemma:
\begin{lemma}\label{lemma}
    Consider an $[[n,k,d]]$ stabilizer code where $d\geq 3$ with a stabilizer group $\mathcal{S}$ and a set of corrections $\{R(\s) : 
\s \in \{0,1\}^{\otimes n-k}\} $. Then there exist two correction operators $R(\s_{1})$ and $R(\s_{2})$, where $\s_{1} \neq \s_{2}$ such that
\begin{equation}
    R(\s_{1}) R(\s_{2}) = \mathcal{L}(\s_{1},\s_{2})R(\s_{1} \oplus \s_{2})
\end{equation}
where $\oplus$ is bitwise addition and $\mathcal{L}(\s_{1},\s_{2}) \in N(S)\setminus \mathcal{S}$, where the function, $\mathcal{L}: \{0,1\}^{\otimes n-k} \times \{0,1\}^{\otimes n-k} \rightarrow N(S)$, is defined such that $\mathcal{L}(\s,\s')= R(\s) R(\s')R(\s \oplus \s') $ for all pairs of syndromes $\s$ and $\s'$.
\end{lemma}
\begin{proof}
    A stabilizer code with distance $d$ corrects all Pauli errors of weight less than $\lfloor{\frac{d-1}{2}}\rfloor$\cite{nielsen2010quantum}. Without loss of generality, the set of corrections can be always chosen to contain this set. Also, there is at least one Pauli operator of weight $\lceil \frac{d+1}{2} \rceil$ that is uncorrectable (or else the code would be a distance $d+1$ code); we denote such an operator as $P^\star$.  For $d\geq3$, the weight of $P^\star$ is also strictly less than $d$.  Note that $P^\star$ can be written as 
\begin{equation}
    P^\star = {P^\star}' {P^\star}'' 
\end{equation}
where both ${P^\star}'$ and ${P^\star}''$ are Pauli operators with weight less than or equal to $\lfloor{\frac{d-1}{2}}\rfloor$.  This means that 
${P^\star}'$ and ${P^\star}''$ are both correctable, meaning they correspond to two correction operators which we can denote $R(\s_{1})$ and $R(\s_{2})$ respectively.   As $P^\star$ has weight $\lceil \frac{d+1}{2} \rceil$, it is \textit{not} correctable, which implies that it does not correspond to a correction operator. Instead, applying Proposition \ref{prop}, $P^\star = L R(\s^\star)$, where $L \in N(S)\setminus \mathcal{S}$. The operator, $L$, cannot be a stabilizer because errors that are stabilizer equivalent to a correction are correctable. Because the syndromes of operators are additive, $\s^\star = \s_{1} \oplus \s_{2}$. This concludes the proof of the lemma.
\end{proof}

Now we can prove Theorem \ref{thm: stab}.
\begin{proof}
    Consider an $[[n,k,d]]$ stabilizer code. Let $\{g_\alpha, \alpha = 1, 2, \dots, n-k\}$ be a set of $n-k$ stabilizer generators for this code. Consider a set of corrections, $\mathfrak{Corrections}= \{R(\s)  : 
\s \in \{0,1\}^{\otimes n-k}\} $, such that for $R(s)\in \mathfrak{Corrections}$, $ R(s) \in  \mathcal{P}_n$, and $R(s)g_i = (-1)^{s_{i}}g_{i}R(s)$ for all $i$. Consider a noisy round of syndrome extraction represented by the map, $\widetilde{\mathcal{R}}$ (see Equation \ref{eq: general noisyrecov}), and $\Pi_{\s}$ (Equation \ref{eq: syn proj}). Let $\Omega$ be the gadget retraction for the code and set of corrections. We will demonstrate Equation \ref{eq: composability violation} by calculating each side of the equality. Calculating $\Omega[\widetilde{\mathcal{R}}]$,

\begin{equation}
    \Omega[\widetilde{\mathcal{R}}] \equiv 
    \mathcal{D}\circ \widetilde{\mathcal{R}} \circ \mathcal{E} = \mathcal{I}_k ,
\end{equation}
where, $\mathcal{D}= \mathcal{E}^{\dagger}\circ \mathcal{R}$ and $\mathcal{I}_k$ is the identity super-operator over $k$ qubits. This is because the output of $\mathcal{E}$ is always in the codespace and $\widetilde{\mathcal{R}}$ will only apply correctable errors to the resulting state, which will get corrected by $\mathcal{R}$, hence the logical qubit is unaffected. Now we calculate $\Omega[\widetilde{\mathcal{R}}\circ\widetilde{\mathcal{R}}]$:
\begin{widetext}
\begin{align} 
    \Omega[\widetilde{\mathcal{R}}\circ\widetilde{\mathcal{R}}] &= \mathcal{D}(\sum_{\s\s'\s''\s'''} C_{\s''' \s''} C_{\s' \s} R(\s''')\Pi_{\s''}R(\s')\Pi_{\s} \mathcal{E}(\cdot)\Pi_{\s}R(\s')\Pi_{\s''}R(\s'''))\\
    &= \mathcal{D}(\sum_{\s\s'\s''\s'''} C_{\s''' \s''} C_{\s' \s} R(\s''')\Pi_{\s''}R(\s')\Pi_{\s} (\Pi_{\vec{0}} \mathcal{E}(\cdot)\Pi_{\vec{0}})\Pi_{\s} R(\s')\Pi_{\s''}R(\s'''))\label{eq:32}\\
    &= \mathcal{D}(\sum_{\s\s'\s''\s'''} C_{\s''' \s''} C_{\s' \s} R(\s''')\Pi_{\s''}R(\s')(\delta_{\s'\vec{0}}\Pi_{\vec{0}} \mathcal{E}(\cdot)\Pi_{\vec{0}})R(\s')\Pi_{\s''}R(\s'''))\label{eq:33}\\
    &= \mathcal{D}(\sum_{\s\s'\s''\s'''} C_{\s''' \s''} C_{\s' \s}\delta_{\s'\vec{0}} R(\s''') (R(\s'')\Pi_{\vec{0}} R(\s''))R(\s')\Pi_{\vec{0}} \mathcal{E}(\cdot)\Pi_{\vec{0}}R(\s')(R(\s'')\Pi_{\vec{0}} R(\s''))R(\s'''))\label{eq:34}\\
    &= \mathcal{D}(\sum_{\s\s'\s''\s'''} C_{\s''' \s'} C_{\s' \vec{0}} \delta_{\s''\s'} \delta_{\s'\vec{0}}R(\s''')R(\s') \Pi_{\vec{0}} \mathcal{E}(\cdot) \Pi_{\vec{0}}R(\s')R(\s''')).\label{eq:35}   
\end{align}
\end{widetext}

 The second line derives from the fact that the output of $\mathcal{E}$ is within the codespace (i.e. $\mathcal{E}= \Pi_{\vec{0}} \mathcal{E}(\cdot) \Pi_{\vec{0}}$). Lines \ref{eq:33}-\ref{eq:35} derive from the fact that
\begin{equation}
    \Pi_{\vec{0}}R(\s')R(\s)\Pi_{\vec{0}}= \delta_{\s\s'}\Pi_{\vec{0}} \quad \forall \s,\s'\in \{0,1\}^{\otimes n-k}.
\end{equation}
This implies that $\Pi_{s}= R(\s)\Pi_{\vec{0}}R(s)$ and $\Pi_{s'}\Pi_{s} = \delta_{\s'\s}\Pi_{s}$. Therefore, 

\begin{align}
    \Pi_{s''}R(s')\Pi_{s}\Pi_{\vec{0}} &= \delta_{\s\vec{0}}\Pi_{s''}R(s')\Pi_{\vec{0}}\\ &=\delta_{\s\vec{0}}R(s'')\Pi_{\vec{0}}R(s'')R(s')\Pi_{\vec{0}}\\ &= \delta_{\s''\s'}\delta_{\s\vec{0}}R(s'')\Pi_{\vec{0}}.
\end{align} \\

These 3 steps correspond to lines \ref{eq:33}, \ref{eq:34},  and \ref{eq:35} respectively. After summing over $s$ and $s''$, we get the following:

\begin{widetext}
\begin{align}
    \Omega[\widetilde{\mathcal{R}}\circ\widetilde{\mathcal{R}}] &= \mathcal{D}(\sum_{\s' \s'''} C_{\s''' \s'} C_{\s' \vec{0}} R(\s''')R(\s') \mathcal{E}(\cdot)R(\s')R(\s'''))\\
    &= \mathcal{D}(\sum_{\s' \s'''} C_{\s''' \s'} C_{\s' \vec{0}} \mathcal{L}(\s''', \s')R(\s'''\oplus \s') \mathcal{E}(\cdot)R(\s'''\oplus \s')\mathcal{L}(\s''', \s'))\\
    &=\sum_{\s' \s'''} C_{\s''' \s'} C_{\s' \vec{0}} P^{\mathcal{L}(\s''', \s')}(\cdot)P^{\mathcal{L}(\s''', \s')}, \label{eq:42}
\end{align}
\end{widetext}
where $\mathcal{L}(\s''', \s'))$ is the logical component of $R(\s')R(\s''')$ per Proposition \ref{prop}, and $P^{\mathcal{L}(\s''', \s')}$ is the $k$ qubit Pauli that corresponds to $\mathcal{L}(\s''', \s')$. Line \ref{eq:42} follows from the fact that the output of $\mathcal{E}$ is always in the codespace, and so the correction components, $R(\s'''\oplus \s')$, will be corrected by the unencoding operation, $\mathcal{D}$. By Lemma \ref{lemma}, one of the $P^{\mathcal{L}(\s''', \s')}$ in the sum will not be the identity, and so the effective logical process will be a Pauli channel not equal to the identity. Therefore,
\begin{equation}
    \Omega[\widetilde{\mathcal{R}} \circ \widetilde{\mathcal{R}}]\neq \mathcal{I}_k = \Omega[\widetilde{\mathcal{R}}] \circ \Omega[\widetilde{\mathcal{R}}].
\end{equation}
\end{proof}

\section{An open systems picture of logical non-Markovianity}\label{sec: causes}
We have found it useful to factor the data qubits that comprise a logical qubit into two subsystems, logical qubits and syndrome qubits.  Because they are coupled by syndrome extraction, the dynamics of the logical qubit[s] are the consequence of a Markov process on the \textit{combined} system, and thus (in general) described by a hidden Markov process. We can describe this in at least two complementary paradigms, and in this section we do just that.  First, we describe syndrome extraction as a Markov chain using stochastic matrices. This paradigm is useful for stabilizer codes with Pauli noise models, initialized into a computational basis state and subsequently measured in the computational basis.  Second, we describe the logical qubits as an open quantum system.  This paradigm can be used for stabilizer codes with more general noise models.  These correspondences allow insights about non-Markovian stochastic processes and quantum non-Markovianity in open quantum systems \cite{li_concepts_2018, milz_quantum_2021} to be applied to logical qubits.
\newline
\subsection{A stochastic picture}

The model we consider in Section \ref{sec: example} can be modeled completely classically as a Markov chain, with the logical qubit dynamics being derivable from the resulting stochastic process. To define a Markov chain, we will need to define the state space of the process, and the transition matrix, which determines how the process evolves from timestep to timestep. Because we initialize in either $\ket{\Bar{0}}$ or $\ket{\Bar{1}}$ of the 3-qubit repetition code and suffer only probabilistic bitflips, the state space of the Markov chain is the set of 3-qubit computational basis states. Each timestep corresponds to a noisy round of error correction and so can be represented by the map $\widetilde{\mathcal{R}}$. The transition matrix of the process,  $\widetilde{\mathcal{R}}_{M}$, corresponds to the matrix of transition probabilities between computational basis states, the matrix elements given by

\begin{equation}
    (\widetilde{\mathcal{R}}_{M})_{(ijk),(qrs)} = Tr(\dyad{ijk}\widetilde{\mathcal{R}}(\dyad{qrs})).
\end{equation}
The full matrix is
\begin{widetext}
\begin{equation}
    \widetilde{\mathcal{R}}_{M} = \ \begin{blockarray}{ccccccccc}
000& 111 & 001 & 110 & 100& 011 & 010 & 101  \\
\begin{block}{(cccccccc)c}
  (1-p)^2 & 0 & (1-p)^2 & 0 & (1-p)^2 & 0 & (1-p)^2 & 0 & \ 000 \\
  0 & (1-p)^2 & 0 & (1-p)^2 & 0 & (1-p)^2 & 0 & (1-p)^2& \ 111 \\
  p(1-p) & 0 & p(1-p) & 0 & 0 & p(1-p) & 0 & p(1-p) &\ 001 \\
  0 & p(1-p) & 0 & p(1-p) & p(1-p) & 0 & p(1-p) & 0 & \ 110 \\
  p(1-p)& 0 & 0 & p(1-p) & p(1-p) & 0 & 0 & p(1-p) & \ 100 \\
  0 & p(1-p) & p(1-p) & 0 & 0 & p(1-p) & p(1-p) & 0 & \ 011 \\
  p^2 & 0 & 0 & p^2 & 0 & p^2 & p^2 & 0 &\ 010 \\
  0 & p^2 & p^2 & 0 & p^2 & 0 & 0 & p^2 & \ 101 \\
\end{block}
\end{blockarray}.
\end{equation}
\end{widetext}
The columns and rows are labeled by their corresponding computational basis states. The labels are further ordered by which syndrome space they belong to (e.g. $000$ and $111$ are in the codespace, whereas $001$ and $110$ are in the $s=(0,1)$ space)  to highlight the syndrome dependence of the process. Having written down the transition matrix for the the average of the noisy error correction round, one can also study the effect that syndrome errors have on the process by decomposing  $\widetilde{\mathcal{R}}$ into maps conditioned on which ancilla errors occur. Let $\widetilde{\mathcal{R}}_{e_{1} e_{2}}$ be the recovery map where an error occurs on one or more ancilla qubits during the syndrome extraction process, where $e_i = 1$ if a bitflip occurred on ancilla $i$ and $e_{i}=0$ otherwise. (Note that $\widetilde{\mathcal{R}}_{00} = \mathcal{R}$). Using the definition of $\widetilde{R}$ from Eq. \ref{eq: noisyrecov},
\begin{equation}
    \widetilde{\mathcal{R}}_{e_1 e_2}(\rho) = \sum_{\s \oplus \s' = (e_{1}, e_{2})} \chi_{\s, \s'} R(\s') \Pi_{\s} \rho \Pi_{\s} R(\s'),
\end{equation}
where $\oplus$ is bitwise addition and $\widetilde{\mathcal{R}}(\rho) = (1-p)^2 \mathcal{R}(\rho)+ p(1-p)(\widetilde{\mathcal{R}}_{10}(\rho) + \widetilde{\mathcal{R}}_{01}(\rho)) + p^2 \widetilde{\mathcal{R}}_{11}(\rho)$. The transition matrix can be decomposed in a similar way. 

This Markov chain description can be used to derive logical dynamics.  Each computational basis state of the data qubits (e.g., $\ket{000}$) corresponds to a joint state of (1) the logical qubit and (2) the syndrome qubits. When the data qubits' state is in the codespace (e.g.~$\ket{\Bar{0}}$ or $\ket{\Bar{1}}$), the state assigned to the logical qubit is obvious (e.g.~$\ket{0}$ or $\ket{1}$, respectively).  For states \textit{outside} the codespace, the state assigned to the logical qubit is determined by the decoding map. So the logical qubit is jointly defied by both the code \textit{and} the decoder.  The decoder defined in Eq. \ref{eq: rep corrections} associates, e.g., $\ket{101}$ with logical $\ket{\Bar{1}}$, etc. Finally, to any computational basis state of the data qubits, we associate the syndrome qubit state corresponding to the outcomes of perfect syndrome extraction.  So, e.g., $\ket{101}$ is associated with the syndrome qubit state of $\ket{11}$ because both code stabilizer generators have eigenvalue $-1$.
We now have a description of a Markov chain and its effect on logical and syndrome qubit information. All of this information can visualized on a cube graph, as shown in Fig.~\ref{fig:cube}:

\begin{figure}[H]
    \centering
    \includegraphics[width=0.8\linewidth]{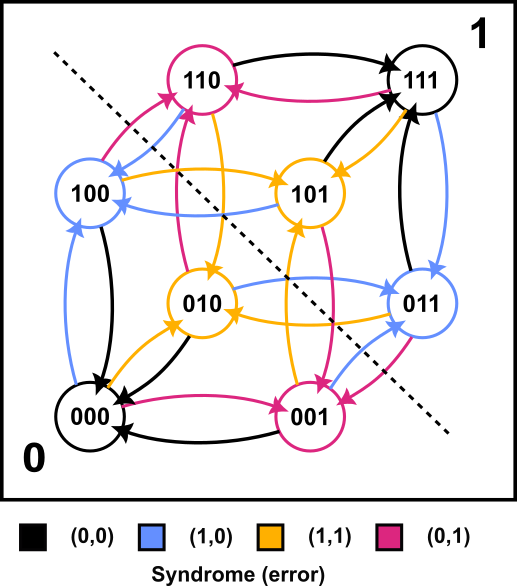}
    \caption{A graph representing the action of 1 noisy QEC round consisting of syndrome extraction and correction in the 3-qubit code. The nodes correspond to computational basis states and the arrows correspond to the action of the noisy recovery map on the states given a particular error, each of which is denoted by a color (e.g. blue corresponds to a bitflip on the first ancilla qubit during syndrome extraction, the blue arrows representing the map $\widetilde{\mathcal{R}}_{10}$). In the absence of syndrome extraction error (the black arrows), the action of the QEC round is to ``cool" the system to the codespace, but syndrome errors will cause the experimenter to confuse one syndrome space for another. A QEC round with syndrome extraction error $(i,j)$ will cool the system to the subspace with syndrome $(i,j)$ via the map $\widetilde{\mathcal{R}}_{ij}$. This leads to logical transition depending on the syndrome of the input state.}
    \label{fig:cube}
\end{figure}
Here, the transitions not displayed map the states to themselves (i.e. $\mathcal{R}(\dyad{000})= \dyad{000} $, and $\widetilde{\mathcal{R}}_{11} (\dyad{101}) =\dyad{101}$).
 This graph can be turned into a Markov chain diagram by labeling each arrow with transition probabilities corresponding to the probability of the corresponding syndrome error (i.e. $p^2$ for the orange arrows which correspond to a syndrome error on both bits). We can partition the physical states into the four states that decode to ``$0$'' (bottom left), and the four that decode to ``$1$'' (top right). Transitions between these sets of states are logical transitions. We see that the same map can cause different logical transitions depending on the value of the syndrome. For example, when $\mathcal{R}_{11}$ is applied to a state with a trivial syndrome or the syndrome $(1,1)$, no logical transitions occur, but when applied to states with the syndromes $(0,1)$ or $(1,0)$ it leads to a logical transition. This means that this stochastic process describes \textit{correlated} dynamics between syndrome and logical bits, which explains why non-Markovianity is observed (as demonstrated in the previous section).

\begin{figure}
    \centering
    \includegraphics[width=1\linewidth]{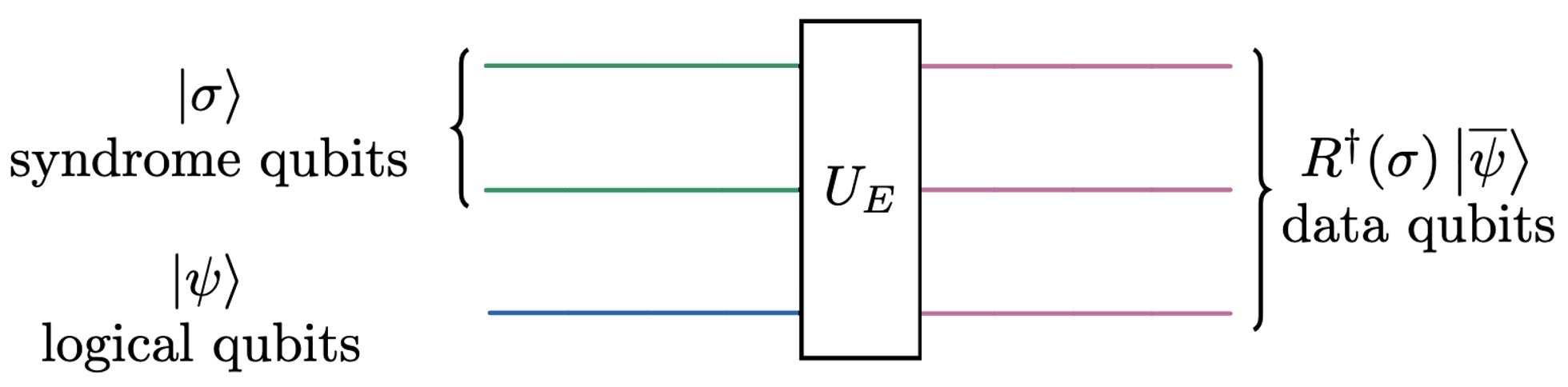}
    \caption{The encoding unitary for the 3-qubit repetition code. The encoding unitary, denoted $U_E$ relates states of the subsystems of syndrome qubits and logical qubits to an encoded state of data qubits. The green, blue, and purple wires correspond to the syndrome qubits, logical qubits and data qubits respectively. When the syndrome qubits are in the all-zeroes state, the corresponding data qubit states are in the codespace, (i.e. for the 3-qubit repetition code, and an arbitrary logical qubit state, $\ket{\psi}$, $U_E\ket{00}\otimes\ket{\psi}=\ket{\Bar{\psi}}$. For a non-trivial syndrome $\s$, $U_E$ maps the syndrome qubit state $\ket{\s}$ to a subspace determined by a choice of Pauli correction, $R(\s)$, such that for an arbitrary logical qubit state $\ket{\psi}$, $U_E\ket{\s}\otimes\ket{\psi}=R^{\dagger}(\s)\ket{\Bar{\psi}}$ (see Table \ref{tab:encunitary}).}
    \label{fig:repunitary}
\end{figure}
\subsection{Logical qubits as an open quantum system}\label{sec: logical oqs}

Modeling logical dynamics as a stochastic process, as in the previous section, is suitable only for stochastic noise, but we expect data qubits to be subject to more general physical error processes.  These general error processes constitute open quantum system dynamics on the logical qubits. For an $[[n,k]]$ code, if $\mathcal{H}_{data}$ is the $2^n$-dimensional space of the data qubits, and $\mathcal{H}_L$ is the $2^k$-dimensional space of logical qubits, then $\mathcal{H}_{data}$ is isomorphic to $ \mathcal{H}_{Syn}\otimes\mathcal{H}_L$, where $\mathcal{H}_{Syn}$ contains the syndrome information, the information necessary to detect errors. We call the frame where logical qubits and syndrome qubits are localized to their own individual subsystems the \emph{logical frame}. We call this isomorphism the \emph{encoding unitary}, and it is determined by both the code and choice of corrections. 

Let $\s \in \{0,1\}^{\otimes n-k}$, be an arbitrary syndrome and $\ket{\psi}$ be an arbitrary $k$-qubit computational basis state in $\mathcal{H}_L$. The encoding unitary can be defined by its action on the computational basis states of $ \mathcal{H}_{Syn}\otimes\mathcal{H}_L$ (see Table \ref{tab:encunitary} for the 3-qubit repetition code). If $R(\s)$ is the correction that corresponds to the syndrome $\s$ (where we also assume $R(\s)$ has the syndrome $s$),   then the action of the encoding unitary, $U_E$, on $\ket{\s}\otimes\ket{\psi}$ is:

\begin{align}
\begin{split}
     U_E: &\mathcal{H}_{Syn}\otimes \mathcal{H}_L \rightarrow \mathcal{H}_{data}\\
        &\ket{\s}\otimes\ket{\psi} \mapsto R^{\dagger}(\s)E \ket{\psi}.\label{eq: enc unitary body}
\end{split}
\end{align}

We prove that this is a valid unitary map in Appendix \ref{appendix: encoding unitary}. Fig \ref{fig:repunitary} illustrates this for the $[3,1,3]$ code, whose corrections are given in Equation \ref{eq: rep corrections}. This definition has two important properties that are conceptually and operationally motivated: (1) When the syndrome qubits are in the state $\s$ the encoding unitary maps the system to a data qubit state with syndrome $\s$. For example, if the syndrome state is the all-zeros state the encoding unitary will map logical qubit states to the codespace. This condition alone does not single out the above definition. For this we need condition (2): the outcome of an ideal logical measurement gadget is only determined by $\mathcal{H}_L$.  An ideal logical measurement gadget can be represented by the POVM $\Bar{M}:= \{ \mathcal{D}^{\dagger}(M_{i}) : i = 1, \dots , n_{m}\} $ where $\mathcal{D}$ is the unencoding operation defined using the same set of corrections as $U_E$ and  $M= \{ M_i : i = 1, \dots , n_{m}\}$ is a POVM defined over $\mathcal{B}(\mathcal{H}_{L})$. One can demonstrate that the definition of $U_E$ implies,
\begin{equation}
    Tr_{Syn}(U_{E}^{\dagger} \rho U_{E}) =\mathcal{D}(\rho),
\end{equation}
where $Tr_{Syn}$ is the partial trace over $\mathcal{H}_{Syn}$.  This means that measuring $\Bar{M}$ over the data qubit space is equivalent to measuring $M$ over the logical qubit space. The two conditions also imply that the gadget retraction, $\Omega$ corresponds to the effect that a process has on the logical qubits if the initial state is within the codespace.

\begin{figure*}
    \centering
    \includegraphics[width=0.85\linewidth]{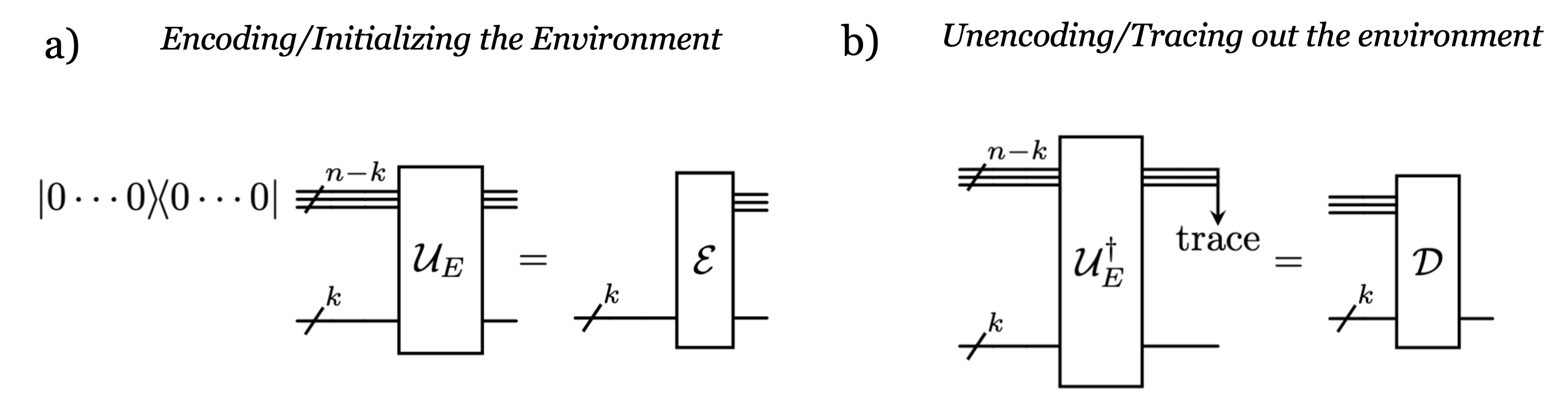}
    \caption{Encoding and decoding in the logical frame and the data qubit frame. Using the definition of the encoding unitary, $U_E$ from Fig \ref{fig:repunitary} and Equation \ref{eq: enc unitary body}, the corresponding superoperator is  $\mathcal{U}_{E}(\cdot) := U_{E}(\cdot)U_{E}^{\dagger}$. a) Applying the encoding operation to a logical qubit state is the same as appending an environment of syndrome qubits in the all zero state in the logical frame. b) Unencoding is the same as tracing out the syndrome qubit environment in the logical frame. }
    \label{fig:encoding}
\end{figure*}

For a process $\Lambda \in \mathcal{B}(\mathcal{B}(\mathcal{H}_{data}))$,

\begin{equation}
    \Omega[\Lambda](\rho) = Tr_{Syn}(U_{E}^{\dagger} \Lambda( U_{E}(\rho\otimes\ket{\vec{0}}\bra{\vec{0}})U_{E}^{\dagger} )U_{E}).
\end{equation}
Another consequence of this definition is that an ideal round of QEC corresponds to resetting the syndrome qubits to the all zero state. With this isomorphism established, at the cost of computing $U_E$, one can port Markovianity conditions from open quantum systems literature to understand logical qubit dynamics. Non-Markovianity of logical qubits happens for the same reason it does for qubits in open quantum systems: they are coupled to an environment ($\mathcal{H}_{Syn}$) that carries memory of the process on the logical qubits and feeds back on the system \cite{li_concepts_2018}.

There are some caveats to this factorization picture of logical processes. The logical frame depends on the ability to perform a perfect state preparation in the codespace and a perfect final measurement gadget. Perfect logical state preparations are analogous to the assumption, in open quantum systems, that the system is initially prepared in a state that is uncorrelated from the environment, (see Appendix \ref{appendix:logical dynamics} and \ref{sec: sufcond} for more details). It is known that if process tomography is performed when initial correlations are present between the system and environment, then there are situations where the dynamics must either be non-linear or non-CP \cite{pechukas1994reduced,alicki1995comment,pechukas1995pechukas, modi2010role} (see the authors' discussion of this in \cite{milz_quantum_2021} ). Another consequence of the argument used in Theorem \ref{thm: stab} is that logical state preparation errors are likely to cause correlations between the logical and syndrome qubits, and so there is the possibility that non-Markovian effects could arise from the SPAM error.

\section{Discussion}

The preceding analysis and examples clearly demonstrate that (and how) Markovian physical errors can give rise to non-Markovian logical errors in non-fault-tolerant syndrome extraction circuits. To do this we have defined an operationally inspired ``button-theoretic" definition of Markovianity, where every time a button with label ``gate $a$" is pressed the resulting process is described by the same CPTP map defined over a the logical qubit space. For $d=3$ codes with no SPAM errors, we found that the logical error rate is initially 0 for 1 round, then jumps to a positive value, and then relaxes rapidly toward an asymptotically constant error rate.  Non-exponential decay of polarization indicates non-Markovianity, but the detailed analysis suggest that non-Markovianity is strongest early on, and disappears rapidly as multiple rounds of error correction are performed.

The source of the observed non-Markovianity is the tendency for imperfect syndrome extraction to leave the data qubits outside the codespace, where the encoded information is less protected from error. Indexing these ``syndrome spaces'' by their syndromes produces a simple picture in which the data qubits' joint Hilbert space factors into logical-qubit and syndrome subsystems, and the syndrome subsystem contains untracked classical information that can become correlated with -- and thus induce non-Markovian dynamics within -- the logical qubit[s].  So logical qubit[s] can be associated with a factor space (whose complement is comprised of syndrome qubits), where the factorization is determined by a complete set of correction operators. The logical qubit[s] are noiseless subsystems of $\mathcal{N}\circ\mathcal{R}$ where $\mathcal{R}$ is the recovery channel determined by the set of corrections, and $\mathcal{N}$ is any channel whose Kraus operators are proportional to the corrections \cite{blume2010information}. Processes on the data qubits can be seen as coupling the logical qubit[s] to the syndrome qubits, and thus defining a hidden Markov model for the non-Markovian dynamics of the logical qubit[s]. This opens up the opportunity to apply Markovianity conditions for open quantum systems to logical processes (such as Markovianity conditions analogous to those in \cite{li_concepts_2018,milz_quantum_2021}).

This work leads to many more questions that could be examined in follow up works. 
For example, it will be valuable to relax the artificial assumption that there are no SPAM errors, and thus to study the Markovianity of logical dynamics where state preparation and measurement gadgets are also noisy. An even more important question is the impact of explicit fault-tolerant circuit design.  Because non-Markovianity does not occur in the absence of syndrome extraction errors, we expect fault-tolerant syndrome extraction to reduce (but not eliminate) non-Markovianity. Also, many candidate platforms for FTQC exhibit other errors such as coherent errors, and non-Markovian \emph{physical} noise such as leakage \cite{hayes_eliminating_2020, babu_state_2021, medford_self-consistent_2013,lucero_high-fidelity_2008}, coupling to two-levels systems \cite{agarwal2024modelling, zhuang2026non}, and $1/f$ noise \cite{yoshihara_decoherence_2006, slichter_measurement-induced_2012, wellstood_low-frequency_1987} among others. A future line of investigation would be to see how physical non-Markovianity affects logical non-Markovianity and its interplay with the emergent effects studied in this paper.

Our results suggest that logical non-Markovianity should be nearly ubiquitous, albeit perhaps very small in magnitude.  Proving bounds on logical non-Markovianity for reasonable noise models and codes would delineate the circumstances under which logical processes are \textit{approximately} Markovian.  This would justify the use of QCVV protocols that assume Markovianity, such as logical RB and logical GST.  Some initial work in this direction appears in Ref.~\cite{kwiatkowski2025constructing}, where a bound on non-Markovianity is studied in repeated syndrome extraction under Pauli noise models.

One question of concern is whether the reduction of syndrome errors will cause a reduction of non-Markovianity that is on par with the reduction of the total logical error. One can imagine that an approximately Markovian logical process might also be approximately noiseless, so distinguishing the process from a Markovian evolution may be equally as difficult as distinguishing the process from ones without errors. The logical error rate may fluctuate by the same proportion, but because it is smaller in magnitude the process would be considered approximately Markovian by a non-Markovian metric based on operational distinguishability of processes (e.g. Schatten measures \cite{milz_quantum_2021}). This suggests that reducing syndrome errors might make standard QCVV techniques more reliable for logical processes.

\begin{acknowledgements}
    We would like to thank Jordan Hines, Alex Kwiatkowski, Andrew Landahl, Cole Maurer,  Mohan Sarovar, and Kevin Young for fruitful technical discussions.  Sandia National Laboratories is a multimission laboratory managed and operated by National Technology \& Engineering Solutions of Sandia, LLC, a wholly owned subsidiary of Honeywell International Inc., for the U.S. Department of Energy’s National Nuclear Security Administration under contract DE-NA0003525.  This paper describes objective technical results and analysis.  The research is based upon work supported in part by the Office of the Director of National Intelligence (ODNI), Intelligence Advanced Research Projects Activity (IARPA), specifically the ELQ program. Any subjective views or opinions that might be expressed in the paper do not necessarily represent the views, official policies or endorsements, either expressed or implied, of the ODNI, IARPA, the U.S. Department of Energy, or the U.S. Government. The U.S. Government is authorized to reproduce and distribute reprints for Governmental purposes notwithstanding any copyright annotation thereon.
\end{acknowledgements}

\appendix
\section{Effective logical dynamics}\label{appendix:logical dynamics}
In this section, we will define all the mathematics necessary to derive the dynamics of logical qubits from their underlying physical processes. To do this we will borrow some of the ideas from \cite{rahn_exact_2002, beale_efficiently_2021} with some slight differences in notation and scope. Consider an $[[n,k]]$ code that encodes $k$ logical qubits (whose Hilbert space we will denote $\mathcal{H}^k$) into $n$ physical or data qubits (whose Hilbert space we will denote $\mathcal{H}^n$). An implementation of an error-corrected logical process using this code is characterized by two things: an encoding operation $\mathcal{E}: \mathcal{B}(\mathcal{H}^k) \rightarrow \mathcal{B}(\mathcal{H}^n)$, and a set of corrections, $\mathfrak{Corrections}$, which will satisfy certain conditions. The set of corrections allows us to define an unencoding operation, $\mathcal{D}: \mathcal{B}(\mathcal{H}^n) \rightarrow \mathcal{B}(\mathcal{H}^k)$, which maps physical states to logical qubit states. Finally, we will define the logical gadget retraction $\Omega$ in terms of these operations. 
\subsubsection{The encoding operation}
The map $\mathcal{E} $ is defined via a linear map $E:\mathcal{H}^k \rightarrow \mathcal{H}^n$
\begin{equation}
    \mathcal{E}(\cdot) := E (\cdot)E^\dagger.
\end{equation}
The operator $E$ is called an encoding isometry, because it must map a basis of logical qubit states to a basis in the codespace in a one-to-one fashion. Isometry implies that it preserves the value of the Hilbert space inner product. Let $\mathcal{C}$ be the codespace of the code, then $Image(E)=\mathcal{C}$ and:

\begin{equation}
    \braket{\psi}{\phi}= \bra{ \psi}E^{\dagger}E\ket{ \phi} \quad \forall \ket{\psi},\ket{\phi} \in \mathcal{H}^{k}.
\end{equation}

Equivalently, $E^{\dagger} E = I_k$, the identity in $\mathcal{B}(\mathcal{H}^{k})$. The operator   $E E^{\dagger}$ is a projector onto the codespace. The encoding isometry of a code is not unique, but all choices for $E$ are equivalent up to isomorphism (unitaries before and after application of $E$). One example of an encoding isometry in the repetition code acts on the computational basis states in the following way:
\begin{equation}
    \ket{0}\mapsto \ket{000},
\end{equation}
\begin{equation}
    \ket{1}\mapsto \ket{111}.
\end{equation}
In general we denote $E\ket{\psi}$ as $\ket{\Bar{\psi}}$. 
\subsubsection{The set of corrections and unencoding}
The set of corrections, $\mathfrak{Corrections}$, represents the decoding strategy of the code, of which there are multiple choices. The main restriction on the set is that for all $R$
\begin{equation}
    \Pi_R = R E E^{\dagger} R^{\dagger}
\end{equation}

is a projector, meaning $\forall R, R' \in \mathfrak{Corrections}$ $\Pi_R \Pi_{R'}= \delta_{R,R'} \Pi_R$. Additionally, the set of projectors should be complete:
\begin{equation}
    \sum_{R \in \mathfrak{Corrections}} \Pi_R = I_n
\end{equation}

The corrections should also correspond to likely errors. This immediately suggests a strategy for error correction. First, measure which syndrome subspace the state is in, and apply the corresponding correction. When the measurement outcomes are averaged over, this process corresponds to the recovery map, $\mathcal{R}$:
\begin{equation}\label{eq: recov map}
    \mathcal{R}(\cdot) = \sum_{R \in \mathfrak{Corrections}} R \Pi_R (\cdot) \Pi_R R^{\dagger}
\end{equation}
The recovery map also allows us to relate general physical qubits states to their corresponding logical qubits via the unencoding operation,
\begin{equation}
    \mathcal{D} = \mathcal{E}^{\dagger}\circ\mathcal{R},
\end{equation}
where $\mathcal{E}^{\dagger} (\cdot) = E^\dagger (\cdot) E$. The unencoding operation associates a physical qubit state with the logical information it has upon ideal decoding. This is partially justified by how logical qubits would be measured in practice: in a fault tolerant scheme, any logical measurement would be performed by a measurement gadget that performs error-correction during the measurement process. Application of $\mathcal{D}$ followed by a measurement of the logical qubit will give the same measurement outcomes as an ideal implementation of such a measurement gadget once the syndrome measurement outcomes are averaged over, as in Eq. \ref{eq: recov map}.

Now that we can represent encoding and unencoding, we can represent the effect that a general physical process has on the logical qubits. Given a choice of encoding $\mathcal{E}$ and decoding via a set of corrections, $\mathfrak{Corrections}$, we define the gadget retraction,
\begin{equation}
    \Omega[\mathcal{N}]:=  \mathcal{D} \circ \mathcal{N} \circ \mathcal{E}
\end{equation}

where $\mathcal{N}$ is a CPTP map in $\mathcal{B}(\mathcal{H}^{n})$. If $\mathcal{N}$ represents a process on the physical qubits, then the gadget retraction maps $\mathcal{N}$ to it's effective logical process $\Omega[\mathcal{N}]$.

\subsection{An open systems correspondence: the encoding unitary}\label{appendix: encoding unitary}
As discussed in section \ref{sec: logical oqs}, for stabilizer codes, processes on physical qubits can be thought of as processes on both the logical qubits and an environment of syndrome qubits. The change in frame is mathematically represented by an \emph{encoding unitary}. The purpose of this section is to mathematically define the encoding unitary for any $[[n,k]]$ stabilizer code. Any encoding unitary, $U_E$, should map the trivial syndrome to the codespace and so, should have the property

\begin{equation}
    U_E(\ket{0}^{n-k} \otimes \ket{\psi}) =  \ket{\Bar{\psi}}\equiv E \ket{\psi}.\label{eq:isometric}
\end{equation}

 Where $\ket{\psi}$ is an arbitrary logical state ($\ket{\psi} \in \mathcal{H}_L\cong \mathcal{H}^k$), and $\ket{0}^{n-k}$ is a syndrome state ($\ket{0}^{n-k}  \in \mathcal{H}_{Syn}\cong \mathcal{H}^{n-k}$) representing the trivial syndrome, and $\ket{\Bar{\psi}} \in \mathcal{H}_{data} \cong \mathcal{H}^n$.  To define $U_E$, it suffices to define a bijection between orthonormal bases of $\mathcal{H}_{data}$ and $\mathcal{H}_{Syn}\otimes \mathcal{H}_L$. Because $U_E$ is a linear operator, we can define it implicitly by showing how the unitary transforms computational basis states. Let $\s \in \{0,1\}^{\otimes n-k}$, be an arbitrary syndrome and $\ket{\psi}$ be an arbitrary $k$ qubit computational basis state in $\mathcal{H}_L$ then the action of $U_E$ on $\ket{\s}\otimes\ket{\psi}$ is:

\begin{align}
\begin{split}
     U_E: &\mathcal{H}_{Syn}\otimes \mathcal{H}_L \rightarrow \mathcal{H}_{data}\\
        &\ket{\s}\otimes\ket{\psi} \mapsto R^{\dagger}(\s)E \ket{\psi}.\label{eq:unitary}
\end{split}
\end{align}

As in the previous section, we are assuming each $\s$ has a corresponding unitary correction $R(\s)$ and the set of corrections, $\{R(\s) : 
\s \in \{0,1\}^{\otimes n-k}\} $ satisfy the condition:
\begin{equation}
    \Pi_{\vec{0}}R^{\dagger}(\s')R(\s)\Pi_{\vec{0}}= \delta_{\s\s'}\Pi_{\vec{0}} \quad \forall \s,\s'\in \{0,1\}^{\otimes n-k}.
\end{equation}

To make this construction general one also has to consider codes where certain syndromes are only used to detect errors and not for error correction. To make $U_E$ a valid unitary on the entire space, for each syndrome pick some Pauli consistent with the syndrome outcome and treat that as a ``correction". This allows for a consistent definition of the logical qubits while still allowing for a description of error detection and post-selection. To prove that this definition of $U_E$ is in fact unitary we need to demonstrate that it maps a complete orthonormal states in one space to a complete orthonormal states in the other. 
\begin{proof}
    Consider two computational basis states in $\mathcal{H}_{Syn}\otimes \mathcal{H}_L$, $\ket{\s}\otimes\ket{i}$ and $\ket{\s'}\otimes\ket{j}$.  Note that $\Pi_{\vec{0}} E \ket{\phi} = E \ket{\phi}  \forall \ket{\phi} \in \mathcal{H}_L $. The inner product after transforming by $U_E$ is

    \begin{gather}
        (\bra{\s}\otimes\bra{i}) U_{E}^{\dagger} U_E (\ket{\s'}\otimes\ket{j})=\bra{i}E^{\dagger}R^{\dagger}(\s)R(\s')E\ket{j}\\
                                                    = \bra{i}E^{\dagger}\Pi_{\vec{0}}R^{\dagger}(\s)R(\s')\Pi_{\vec{0}}E\ket{j} \\
                                                      = \bra{i}E^{\dagger} \delta_{\s\s'}\Pi_{\vec{0}}E\ket{j}\\ 
                                                      = \delta_{\s\s'}\bra{i}E^{\dagger} E\ket{j}\\
                                                      = \delta_{\s\s'}\delta_{ij}
    \end{gather}
    So $U_E$ maps a complete orthonormal basis of $\mathcal{H}_{Syn}\otimes \mathcal{H}_L$ to an orthonormal basis in $\mathcal{H}_{data}$, and because the dimensions of the two spaces are the same, $U_E$ must be unitary.
\end{proof}

 Note, this unitary is similar but not identical with the isomorphism given in \cite{knill_theory_1997, knill_theory_2000}. The authors in \cite{carrozza2024correspondence} also prove that there is a general one-to-one correspondence between tensor factorizations of Hilbert space and maximal correctable error sets, and relate QECs to the theory of reference frames \cite{bartlett_reference_2007}.

A simple example of the encoding unitary is in the 3-qubit repetition ($[3,1,3]$) code. Here, $\mathcal{S}= \bigl \langle ZZI, IZZ \bigr \rangle$, $\mathcal{C} = 
\text{Span}(\{\ket{000}, \ket{111}\})$, and the code can correct any single bit flip error.

\begin{table}[]\label{table: rep code}
    \centering
    \begin{tabular}{|c | c | c|} 
        \hline
        Logical state  & Syndrome & Encoded State \\
        $\ell$&$\s$&$U_E(\ket{\s}\otimes \ket{\ell})$\\
        \hline
        \hline
        $0$&$00$ & $\ket{000}$\\
        \hline
        $0$& $01$& $\ket{001}$\\
        \hline
        $0$& $10$& $\ket{100}$\\
        \hline
        $0$& $11$& $\ket{010}$\\
        \hline
        $1$&$00$ & $\ket{111}$\\
        \hline
        $1$&$01$ & $\ket{110}$\\
        \hline
        $1$& $10$& $\ket{011}$\\
        \hline
        $1$&$11$ & $\ket{101}$\\
        \hline
    \end{tabular}
    \caption{A table of the encoding unitary for the 3-qubit repetition code. Shown is how it maps the 3-qubit computational basis states to the basis states of logical and syndrome qubits. }
    \label{tab:encunitary}
\end{table}

Note that the encoding unitary, as defined in Eq. \ref{eq:unitary}, is not the only unitary with the property in Eq. \ref{eq:isometric}, but we would argue that this is the most natural definition of the encoding unitary, because it corresponds closely to how logical qubits are measured in on a real quantum computer. In practice, when the logical qubits are measured, there is always some error-correction that occurs. A consequence of the definition of the encoding unitary is
\begin{equation}
    Tr_{Syn}(U_{E}^{\dagger} \rho U_{E}) = E^{\dagger}(\mathcal{R}(\rho))E =\mathcal{E}^{\dagger}(\mathcal{R}(\rho)).
\end{equation}
This means that the marginal state of the logical qubits given by tracing out the syndrome qubits is the same the one given by applying the unencoding operation to the data qubit state. This means that direct measurement of the logical qubits in the logical frame corresponds to an ideal measurement gadget of a logical operator in the physical frame. Thus, our definition of the encoding unitary folds QEC into the definition of the logical qubits and measurements themselves rather than QEC being thought of as a process that occurs before measurement. With this definition of logical qubits in mind, if the syndrome qubits are initialized in the all zero state, followed by the application of  a process $\Lambda$ to the data qubits in the physical frame, then the resulting logical process is given by the gadget retraction:

\begin{equation}
    \Omega[\Lambda](\rho_L) = Tr_{Syn}(U_{E}^{\dagger} \Lambda( U_{E}(\rho_{L}\otimes\ket{\vec{0}}\bra{\vec{0}})U_{E}^{\dagger} )U_{E}).
\end{equation}
One can generalize this map when the syndrome qubits are initialized in a different state, $\rho_{Syn}\in \mathcal{B}(\mathcal{H}^{n-k})$. The generalized gadget retraction is given by:
\begin{equation}
    \Omega_{\rho_{Syn}}[\Lambda](\rho_{L}) := Tr_{Syn}(U_{E}^{\dagger} \Lambda( U_{E}(\rho_{L}\otimes\rho_{Syn})U_{E}^{\dagger} )U_{E}).
\end{equation}
This can also be written as,
\begin{equation}
    \Omega_{\rho_{Syn}}[\Lambda](\rho_{L}) := \mathcal{D}\circ\Lambda\circ\mathcal{E}_{\rho_{Syn}}(\rho_{L}),
\end{equation}
where 

\begin{equation}\label{eq:gen enc op}
    \mathcal{E}_{\rho_{Syn}}(\cdot):= U_{E}((\cdot)\otimes\rho_{Syn})U_{E}^{\dagger}
\end{equation}
is the generalized encoding operation. 
\begin{figure}
    \centering
    \includegraphics[width=0.75\linewidth]{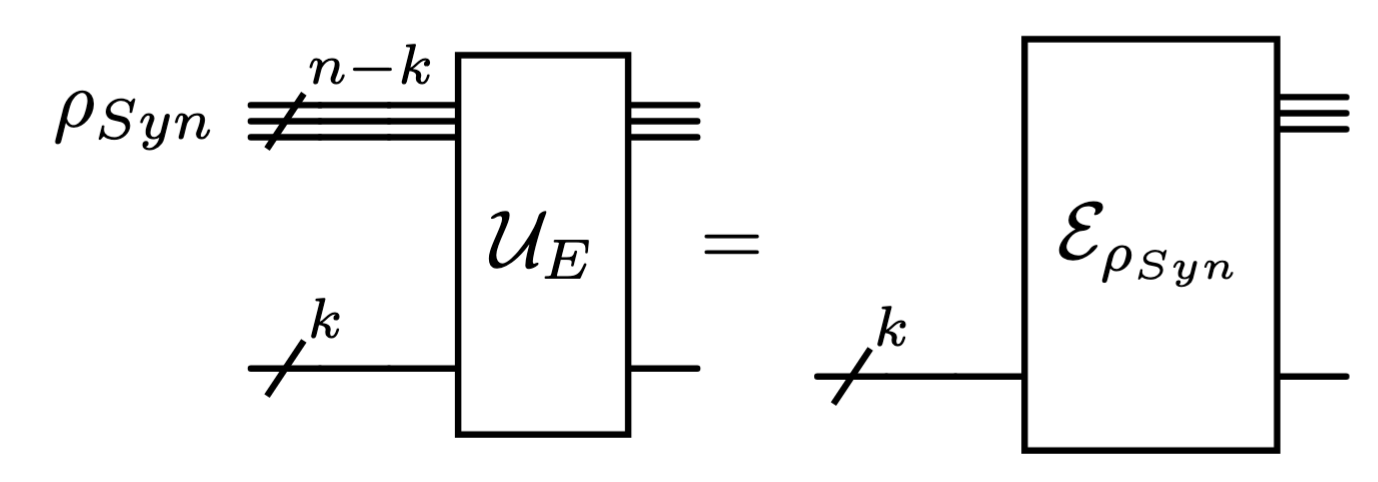}
    \caption{The generalized encoding operation in the logical frame. The generalized encoding operation, $\mathcal{E}_{\rho_{Syn}}$, is the same as initializing the syndrome qubits in the state, $\rho_{Syn}$, in the logical frame.}
    \label{fig:enter-label}
\end{figure}

With the formalization of constructing effective logical process under the assumption of perfect state preparation and measurement characterized, we can move on to defining what non-Markovianity means for logical qubits under these conditions.

\section{Logical Markovianity}\label{appendix:LM}
In this section, we give a formal definition of logical Markovianity. The two notions of Markovianity we talk about is button-theoretic Markovianity and gate composability. In these definitions, the quantum computer is abstracted as a black box. Operations such as state preparations, gates, and measurements are abstracted as labeled buttons. \emph{Button-theoretic Markovianity} is the assumption that every time a particular button is pressed, the same operation occurs.

\begin{definition}[Button-theoretic Markovianity]\label{def:Buttons}
    Consider a set of gate set labels, $\mathcal{A}= \{\mathcal{A}_{states}, \mathcal{A}_{gates}, \mathcal{A}_{effects}\}$, where $\mathcal{A}_{states}$, $\mathcal{A}_{gates}$ and $\mathcal{A}_{effects}$ label state preparations, gates, and effects respectively. Consider the set of circuit instructions, $\mathcal{A}^{\star}$ constructed from the labels in $\mathcal{A}$, where
    \begin{multline}
                \mathcal{A}^{\star} =\{(\alpha_{0}, \beta_1, \dots, \beta_{i}, e_f ) : \alpha_{0} \in\mathcal{A}_{states},\ i\geq1\\ \beta_{i'}\in\mathcal{A}_{gates} \  1\leq i' \leq i, \ e_f \in\mathcal{A}_{effects}\}.
    \end{multline}

    A black box quantum computer is a function that maps circuit instructions to probabilities: $\mathcal{Q}:\mathcal{A}^{\star} \rightarrow [0,1]$. Consider a gate set $\mathcal{G}$ defined over a Hilbert space $\mathcal{H}$.
    \begin{equation}
    \mathcal{G} = \bigl\{ \{\rho_i \}_{i=1}^{N_\rho}; \{G_i\}_{i=1}^{N_G} ; \{ E_{i}^{(m)}\}_{m=1,i=1}^{N_{M}, N_{E}^{(m)}} \bigr\},
    \end{equation}
    where $\rho_i$ is a density matrix on some Hilbert space, $\mathcal{H}$, $G_i :\mathcal{B}(\mathcal{H})\rightarrow \mathcal{B}(\mathcal{H})$, and  $E_{i}^{(m)} \in \mathcal{B}(\mathcal{H})$. The symbols $N_\rho$, $N_G$, $N_M$, and $N_{E}^{(m)}$ denote the number of state preparations, gates, measurements, and measurement outcomes for measurement $m$, respectively. The function, $\mathcal{Q}$, is \textbf{button-theoretic Markovian} for $\mathcal{G}$ if and only if there exist bijective functions, 
    \begin{align}
            g_{\mathfrak{s}}: &\mathcal{A}_{states} \rightarrow \{\rho_i \}_{i=1}^{N_\rho},\\
            g_{\mathfrak{G}}: &\mathcal{A}_{gates} \rightarrow \{G_i\}_{i=1}^{N_G},\\
            g_{\mathfrak{e}}:& \mathcal{A}_{effects} \rightarrow \{ E_{i}^{(m)}\}_{m=1,i=1}^{N_{M}, N_{E}^{(m)}},\\
            g &= g_{\mathfrak{s}}\cup g_{\mathfrak{G}}\cup g_{\mathfrak{e}}
    \end{align}
    such that for any sequence of circuit instructions, $\alpha^{\star}$, in $\mathcal{A}^{\star}$ where $\alpha^{\star} = (\alpha_{0}, \beta_1, \dots, \beta_{i}, e_f ) $, 
    \begin{equation}
        \mathcal{Q}(\alpha^{\star}) = Tr(g(e_{f})g(\beta_{i})\circ \cdots \circ g(\beta_{1})[g(\alpha_{0})]).
    \end{equation}
\end{definition}
When Markovian physical gates correspond to Markovian effective logical gates, we call them logical gate composable.
\begin{definition}[Gate Composability]\label{def:gate composability}
Consider a gadget retraction, $\Omega:\mathcal{B}(\mathcal{B}(\mathcal{H}^n))\rightarrow \mathcal{B}(\mathcal{B}(\mathcal{H}^k))$, and a set of gates $\mathcal{F}$. $\mathcal{F}$ is \textbf{logical gate composable} for $\Omega$ if and only if for all $G_{a}, G_{b} \in \mathcal{F}$
\begin{equation}
    \Omega[G_{a} \circ G_{b}] = \Omega[G_{a}]\circ \Omega[G_{b}].
\end{equation}
\end{definition}

\begin{proposition}
    Consider a gate set $\mathcal{G}$ defined over an $n$ qubit Hilbert space $\mathcal{H}^n$.
    \begin{equation}
    \mathcal{G} = \bigl\{ \{\mathcal{E}(\rho_{L,i})\}_{i=1}^{N_\rho}; \{G_i\}_{i=1}^{N_G} ; \{ \mathcal{D}^{\dagger} (E_{L,i}^{(m)})\}_{m=1,i=1}^{N_{M}, N_{E}^{(m)}} \bigr\},
    \end{equation}
    where $\rho_{L,i}\in \mathcal{B}(\mathcal{H}^k)$, $G_i :\mathcal{B}(\mathcal{H}^n)\rightarrow \mathcal{B}(\mathcal{H}^n)$, and  $E_{L,i}^{(m)}\in \mathcal{B}(\mathcal{H}^k)$. The maps $\mathcal{E}$ and $\mathcal{D}$, are the the encoding and unencoding operation for some $[[n,k]]$ code and set of corrections, $\mathfrak{Corrections}$. The symbols $N_\rho$, $N_G$, $N_M$, and $N_{E}^{(m)}$ denote the number of state preparations, gates, measurements, and measurement outcomes for measurement $m$, respectively. Let $\Omega[\cdot]\equiv \mathcal{D}\circ(\cdot)\circ\mathcal{E}$ be the gadget retraction for the code. Consider the gate set
    \begin{equation}
        \mathcal{G}_L = \bigl\{ \{\rho_{L,i}\}_{i=1}^{N_\rho}; \{\Omega[G_i]\}_{i=1}^{N_G} ; \{ E_{L,i}^{(m)}\}_{m=1,i=1}^{N_{M}, N_{E}^{(m)}} \bigr\},
    \end{equation}
    If a black box quantum computer is button-theoretic for a gateset $\mathcal{G}$ then it is also button-theoretic for $\mathcal{G}_L$ if and only if $\{G_i\}_{i=1}^{N_G} $ is gate composable for $\Omega$.
\end{proposition}
This can be demonstrated by simple substitution of Definitions \ref{def:Buttons} and \ref{def:gate composability}.

\subsection{Relation to other Definitions of Quantum Markovianity }

The form of Definitions \ref{def:Buttons} and \ref{def:gate composability} might seem unintuitive at first, but they capture implicit facts about using quantum circuits to model quantum processes. Quantum algorithms can be designed as quantum circuits composed from a fixed set of gates. It is assumed that a when two gates in a circuit diagram have the same label, they correspond to the same quantum process. Suppose there is a gate colloquially called "$G$" that has an implementation that is context-dependent. For example, if the noise channel of "$G$" changes after the 40\textsuperscript{th} time it is applied (e.g. its depolarizing rate quadruples), then ``$G$'' would be represented by two different CPTP maps: one for the first 40 applications of ``$G$'' and one for the 41\textsuperscript{st} onward. If an experimenter, unaware of this context-dependence, estimated ``$G$'' from sequences shorter than 40 and used the resulting estimate as a model for sequences of any length, they would fail to reproduce the behavior of longer sequences. In order to faithfully model processes containing sequences of ``$G$'' longer than 40, those two maps must be known and labeled (e.g. the two variations of ``$G$'' could be labeled as $G_1$ and $G_2$). This labeling is precisely what makes button-theoretic Markovianity distinct from notions of Markovianity designed for open quantum systems like CP-divisibility \cite{Wolf2008-hx,Rivas2014-review, rivas2010entanglement}. The above situation and the example described in Section \ref{sec: example} are both divisible into the composition of CP maps, but are considered non-Markovian due to the unlabeled context-dependence of the gates (See the supplementary material of \cite{vezvaee2026surface,google_quantum_ai_suppressing_2023} for practical examples of CP-divisible models of context-dependent logical processes).

Definitions of quantum Markovianity are also dependent on the operational processes one is interested in, and related methods for witnessing the non-Markovianity. Some definitions of CP-divisibility presuppose dynamics that can be stopped and characterized at any value of a continuous set of times, whereas we only allow for characterization at times before or after the application of a logical gate. Furthermore, the logical gates and measurements we consider always contain a QEC cycle. We differ in this sense from the authors of \cite{endo2025non} who identify the failure of CP-divisibility of a process consisting of data qubit noise and a round of ideal QEC. Witnessing non-Markovianity as they define implies ability to measure the system after the noise but before the QEC cycle, which may be relevant for the operational setting they are concerned with (quantum error mitigation), but not ours (modeling quantum processes with gate sets, where each gate implicitly includes a round of QEC). 

We defined logical gate composability to cut to the heart of \emph{emergent} non-Markovianity. Emergent non-Markovianity means that the system is button-theoretic at the physical level but not at the logical level. In order to lack emergent non-Markovianity, it is not enough for the system to be button-theoretic Markovian at the logical level; in the studied example, because the process is CP-divisible, one could just label the logical idle gate at each timestep with a different name, satisfying button-theoretic Markovianity. To lack emergent non-Markovianity, there must be a correspondence between the logical gates at the physical and logical levels, using the same context labels for both. This is what logical gate composability is designed to capture.

The relationship between logical gate composability and button-theoretic Markovianity can be thought of as analogous to the relationship of open quantum systems concepts of composability or no information backflow (NIB) with CP-divisibility (to see the relation between Markovianity definitions in open quantum systems see \cite{li_concepts_2018}). Button-theoretic Markovianity (CP-divisibility) refers only to the subsystem of interest but logical gate composability (composability and NIB) refers to logical-syndrome (system-environment) interactions. This relationship can be made explicit using the open quantum systems correspondence of logical systems discussed above. Even so, the concepts remain distinct due to the aforementioned problem of labeling. Taking this correspondence seriously, instances of CP-divisible emergent logical non-Markovianity may be related to instances CP-divisible processes that still carry multi-time correlations, which are known to exist in open quantum systems \cite{milz_completely_2019, milz_quantum_2021}. The correspondence of logical qubits with open quantum systems has not only inspired our definitions, but has also inspired sufficient conditions for logical gate composability discussed in the next section.

\section{Sufficient conditions for logical Markovianity}\label{sec: sufcond}
\begin{theorem}\label{thm: sufcond}
Consider a set of gates, $\mathcal{F}$ (where $\forall G \in \mathcal{F} $, $G$ is a linear map $G:\mathcal{B}(\mathcal{H}^n)\rightarrow \mathcal{B}(\mathcal{H}^n)$), a generalized encoding operation, $\mathcal{E}_{\rho}:\mathcal{B}(\mathcal{H}^k)\rightarrow \mathcal{B}(\mathcal{H}^n)$ (see Equation \ref{eq:gen enc op}), a decoding operation $\mathcal{D}:\mathcal{B}(\mathcal{H}^k)\rightarrow \mathcal{B}(\mathcal{H}^n)$ (such that $\mathcal{D}\circ \mathcal{E}_{\rho}= \mathcal{I}_k $), and the corresponding gadget retraction, $\Omega_{\rho}(\cdot):=\mathcal{D}\circ(\cdot)\circ \mathcal{E}_{\rho}$. If $\forall G \in \mathcal{F}$ $Image(G) \subseteq Image(\mathcal{E}_{\rho})$ then $\mathcal{F}$ is logical gate composable for $\Omega_{\rho}$.
\end{theorem}

\begin{proof}
    Consider a set of gates $\mathcal{F}$ such that $\forall G \in \mathcal{F}$ $Image(G) \subseteq Image(\mathcal{E}_{\rho})$. Note that $\mathcal{E}_{\rho} \circ \mathcal{D}$ acts as the identity on operators in $Image(\mathcal{E}_{\rho})$, so considering two arbitrary gates in $\mathcal{F}$, $G_a$, and $G_b$
    \begin{align}
        \Omega_\rho [G_a \circ G_b] & := \mathcal{D}\circ G_a \circ G_b \circ \mathcal{E}\\
                                    & = \mathcal{D}\circ G_a \circ \mathcal{E}_{\rho} \circ \mathcal{D} \circ G_b \circ \mathcal{E}\\
                                    &= \Omega_\rho [G_{a}]\circ \Omega_\rho [G_{b}].
    \end{align}
    The second line follows from the output of each gate being in $Image(\mathcal{E}_{\rho})$. Therefore $\mathcal{F}$ is logical gate composable for $\Omega_\rho$.
\end{proof}
From this, an intuitive corollary can be proven:
\begin{corollary}
    Consider a set of gates, $\mathcal{F}$, an encoding operation $\mathcal{E}$, a decoding operation $\mathcal{D}$, and the corresponding gadget retraction $\Omega$. If every gate in the set is followed by a perfect round of error correction, in other words $\forall G \in \mathcal{F}$ $ G = \mathcal{R}\circ \widetilde{G}$ for some arbitrary CPTP map $\widetilde{G}$, then $\mathcal{F}$ is logical gate composable for $\Omega$.
\end{corollary}
\begin{proof}
    If every gate contains a round of perfect error correction then it's output will be in the codespace. Mathematically, $\forall G \in \mathcal{F}$, $ Image(G) = Image(\mathcal{R}) = Image(\mathcal{E})$. By \ref{thm: sufcond}, $\mathcal{F}$ is logical gate composable for $\Omega$.
\end{proof}
The intuition behind this is that perfect QEC always leaves the data qubits in the codespace, and if the process both starts and ends in the codespace (which is isomorphic to the logical qubit space), then the physical process is isomorphic to a process on the logical qubits. In many fault-tolerant schemes, gate gadgets are alternated with QEC gadgets so, to first order, the process is Markovian in this sense. Only when the error correction gadgets have faults is there a possibility of logical non-Markovianity.

\bibliographystyle{quantum}
\bibliography{NMLQB}

@article{burkard2009non,
  title={Non-Markovian qubit dynamics in the presence of 1/f noise},
  author={Burkard, Guido},
  journal={Physical Review B},
  volume={79},
  number={12},
  pages={125317},
  year={2009},
  publisher={APS},
  doi = {10.1103/PhysRevB.79.125317}
}

@article{white2020demonstration,
  title={Demonstration of non-Markovian process characterisation and control on a quantum processor},
  author={White, Gregory AL and Hill, Charles D and Pollock, Felix A and Hollenberg, Lloyd CL and Modi, Kavan},
  journal={Nature Communications},
  volume={11},
  number={1},
  pages={6301},
  year={2020},
  publisher={Nature Publishing Group UK London},
  doi= {10.1038/s41467-020-20113-3}
}

@article{iyer_small_2018,
	title = {A small quantum computer is needed to optimize fault-tolerant protocols},
	volume = {3},
	number = {3},
	journal = {Quantum Science and Technology},
	author = {Iyer, Pavithran and Poulin, David},
	year = {2018},
	note = {Publisher: IOP Publishing},
	pages = {030504},
    doi= {10.1088/2058-9565/aab73c}
}

@misc{fan2024randomized,
      title={Randomized Benchmarking with Synthetic Quantum Circuits}, 
      author={Yale Fan and Riley Murray and Thaddeus D. Ladd and Kevin Young and Robin Blume-Kohout},
      year={2024},
      eprint={2412.18578},
      archivePrefix={arXiv},
      primaryClass={quant-ph},
      url={https://arxiv.org/abs/2412.18578}, 
}

@article{pechukas1994reduced,
  title={Reduced dynamics need not be completely positive},
  author={Pechukas, Philip},
  journal={Physical review letters},
  volume={73},
  number={8},
  pages={1060},
  year={1994},
  publisher={APS},
  doi= {10.1103/PhysRevLett.73.1060}
}

@article{alicki1995comment,
  title={Comment on “reduced dynamics need not be completely positive”},
  author={Alicki, Robert},
  journal={Physical review letters},
  volume={75},
  number={16},
  pages={3020},
  year={1995},
  publisher={APS},
  doi = {10.1103/PhysRevLett.75.3020}
}

@article{pechukas1995pechukas,
  title={Pechukas replies},
  author={Pechukas, Philip},
  journal={Physical review letters},
  volume={75},
  number={16},
  pages={3021},
  year={1995},
  publisher={APS},
  doi ={10.1103/PhysRevLett.75.3021}
}

@article{milz_quantum_2021,
	title = {Quantum {Stochastic} {Processes} and {Quantum} non-{Markovian} {Phenomena}},
	volume = {2},
	url = {https://link.aps.org/doi/10.1103/PRXQuantum.2.030201},
	doi = {10.1103/PRXQuantum.2.030201},
	number = {3},
	journal = {PRX Quantum},
	author = {Milz, Simon and Modi, Kavan},
	month = jul,
	year = {2021},
	note = {Publisher: American Physical Society},
	pages = {030201},
}

@article{dennis_topological_2002,
	title = {Topological quantum memory},
	volume = {43},
	number = {9},
	journal = {Journal of Mathematical Physics},
	author = {Dennis, Eric and Kitaev, Alexei and Landahl, Andrew and Preskill, John},
	year = {2002},
	note = {Publisher: American Institute of Physics},
	pages = {4452--4505},
    doi = {10.1063/1.1499754}
}

@article{aleferis_quantum_2006,
	title = {Quantum accuracy threshold for concatenated distance-3 code},
	volume = {6},
	journal = {Quant. Inf. Comput},
	author = {Aleferis, P and Gottesman, D and Preskill, J},
	year = {2006},
	pages = {97--165},
    url = {https://dl.acm.org/doi/abs/10.5555/2011665.2011666}
    
}

@misc{ceasura_non-exponential_2022,
	title = {Non-{Exponential} {Behaviour} in {Logical} {Randomized} {Benchmarking}},
	author = {Ceasura, Athena and Iyer, Pavithran and Wallman, Joel J. and Pashayan, Hakop},
	year = {2022},
    eprint= {2212.05488},
    archivePrefix={arXiv},
    primaryClass={quant-ph},
	url= {https://doi.org/10.48550/arXiv.2212.05488},
}

@article{knill_resilient_1998,
	title = {Resilient quantum computation: error models and thresholds},
	volume = {454},
	number = {1969},
	journal = {Proceedings of the Royal Society of London. Series A: Mathematical, Physical and Engineering Sciences},
	author = {Knill, Emanuel and Laflamme, Raymond and Zurek, Wojciech H},
	year = {1998},
	note = {Publisher: The Royal Society},
	pages = {365--384},
    doi = {10.1098/rspa.1998.0166},

}

@article{knill_theory_2000,
	title = {Theory of {Quantum} {Error} {Correction} for {General} {Noise}},
	volume = {84},
	issn = {1079-7114},
	url = {http://dx.doi.org/10.1103/PhysRevLett.84.2525},
	doi = {10.1103/physrevlett.84.2525},
	number = {11},
	journal = {Physical Review Letters},
	author = {Knill, Emanuel and Laflamme, Raymond and Viola, Lorenza},
	month = mar,
	year = {2000},
	note = {Publisher: American Physical Society (APS)},
	pages = {2525--2528},
    doi = {10.1103/PhysRevLett.84.2525},
}

@article{knill_theory_1997,
	title = {Theory of quantum error-correcting codes},
	volume = {55},
	number = {2},
	journal = {Physical Review A},
	author = {Knill, Emanuel and Laflamme, Raymond},
	year = {1997},
	note = {Publisher: APS},
	pages = {900},
    doi = {10.1103/PhysRevA.55.900},
}

@article{iverson_coherence_2020,
	title = {Coherence in logical quantum channels},
	volume = {22},
	number = {7},
	journal = {New Journal of Physics},
	author = {Iverson, Joseph K and Preskill, John},
	year = {2020},
	note = {Publisher: IOP Publishing},
	pages = {073066},
    doi = {10.1088/1367-2630/ab8e5c}
}

@article{pollock_non-markovian_2018,
	title = {Non-{Markovian} quantum processes: {Complete} framework and efficient characterization},
	volume = {97},
	url = {https://link.aps.org/doi/10.1103/PhysRevA.97.012127},
	doi = {10.1103/PhysRevA.97.012127},
	number = {1},
	journal = {Phys. Rev. A},
	author = {Pollock, Felix A. and Rodríguez-Rosario, César and Frauenheim, Thomas and Paternostro, Mauro and Modi, Kavan},
	month = jan,
	year = {2018},
	note = {Publisher: American Physical Society},
	pages = {012127},
}

@article{milz_completely_2019,
	title = {Completely {Positive} {Divisibility} {Does} {Not} {Mean} {Markovianity}},
	volume = {123},
	url = {https://link.aps.org/doi/10.1103/PhysRevLett.123.040401},
	doi = {10.1103/PhysRevLett.123.040401},
	number = {4},
	journal = {Phys. Rev. Lett.},
	author = {Milz, Simon and Kim, M. S. and Pollock, Felix A. and Modi, Kavan},
	month = jul,
	year = {2019},
	note = {Publisher: American Physical Society},
	pages = {040401},
}

@article{rahn_exact_2002,
	title = {Exact performance of concatenated quantum codes},
	volume = {66},
	url = {https://link.aps.org/doi/10.1103/PhysRevA.66.032304},
	doi = {10.1103/PhysRevA.66.032304},
	number = {3},
	journal = {Phys. Rev. A},
	author = {Rahn, Benjamin and Doherty, Andrew C. and Mabuchi, Hideo},
	month = sep,
	year = {2002},
	note = {Publisher: American Physical Society},
	pages = {032304},
}

@article{proctor_what_2017,
	title = {What {Randomized} {Benchmarking} {Actually} {Measures}},
	volume = {119},
	url = {https://link.aps.org/doi/10.1103/PhysRevLett.119.130502},
	doi = {10.1103/PhysRevLett.119.130502},
	number = {13},
	journal = {Phys. Rev. Lett.},
	author = {Proctor, Timothy and Rudinger, Kenneth and Young, Kevin and Sarovar, Mohan and Blume-Kohout, Robin},
	month = sep,
	year = {2017},
	note = {Publisher: American Physical Society},
	pages = {130502},
    doi ={10.1103/PhysRevLett.119.130502}
}

@article{emerson_scalable_2005,
	title = {Scalable noise estimation with random unitary operators},
	volume = {7},
	url = {https://dx.doi.org/10.1088/1464-4266/7/10/021},
	doi = {10.1088/1464-4266/7/10/021},
	number = {10},
	journal = {Journal of Optics B: Quantum and Semiclassical Optics},
	author = {Emerson, Joseph and Alicki, Robert and Życzkowski, Karol},
	month = sep,
	year = {2005},
	pages = {S347},
}

@misc{tanggara2024strategic,
      title={Strategic Code: A Unified Spatio-Temporal Framework for Quantum Error-Correction}, 
      author={Andrew Tanggara and Mile Gu and Kishor Bharti},
      year={2024},
      eprint={2405.17567},
      archivePrefix={arXiv},
      primaryClass={quant-ph},
      url={https://arxiv.org/abs/2405.17567}, 
}

@misc{kwiatkowski2025constructing,
      title={Constructing an approximate logical Markovian model of consecutive QEC cycles of a stabilizer code}, 
      author={Alex Kwiatkowski and Aaron J. Friedman and Shawn Geller and Jalan A. Ziyad and Scott Glancy and Emanuel Knill},
      year={2025},
      eprint={2509.16887},
      archivePrefix={arXiv},
      primaryClass={quant-ph},
      url={https://arxiv.org/abs/2509.16887}, 
}

@article{blume2010information,
  title={Information-preserving structures: A general framework for quantum zero-error information},
  author={Blume-Kohout, Robin and Ng, Hui Khoon and Poulin, David and Viola, Lorenza},
  journal={Physical Review A},
  volume={82},
  number={6},
  pages={062306},
  year={2010},
  publisher={APS},
  doi = {10.1103/PhysRevA.82.062306}
}

@article{fowler2012surface,
  title={Surface codes: Towards practical large-scale quantum computation},
  author={Fowler, Austin G and Mariantoni, Matteo and Martinis, John M and Cleland, Andrew N},
  journal={Physical Review A},
  volume={86},
  number={3},
  pages={032324},
  year={2012},
  publisher={APS},
  doi = {10.1103/PhysRevA.86.032324}
}

@inproceedings{shor1996fault,
  title={Fault-tolerant quantum computation},
  author={Shor, Peter W},
  booktitle={Proceedings of 37th conference on foundations of computer science},
  pages={56--65},
  year={1996},
  organization={IEEE},
  doi ={10.1109/SFCS.1996.548464}

}

@inproceedings{aharonov1997fault,
  title={Fault-tolerant quantum computation with constant error},
  author={Aharonov, Dorit and Ben-Or, Michael},
  booktitle={Proceedings of the twenty-ninth annual ACM symposium on Theory of computing},
  pages={176--188},
  year={1997},
  doi ={10.1145/258533.258579}
}

@misc{knill1996threshold,
  title={Threshold accuracy for quantum computation},
  author={Knill, Emanuel and Laflamme, Raymond and Zurek, Wojciech},
  journal={arXiv preprint quant-ph/9610011},
  year={1996},
  eprint={9610011},
  archivePrefix={arXiv},
  primaryClass={quant-ph},
  doi= {10.48550/arXiv.quant-ph/9610011},
}

@article{kitaev2003fault,
  title={Fault-tolerant quantum computation by anyons},
  author={Kitaev, A Yu},
  journal={Annals of physics},
  volume={303},
  number={1},
  pages={2--30},
  year={2003},
  publisher={Elsevier},
  doi={10.1016/S0003-4916(02)00018-0}
}

@article{steane1996multiple,
  title={Multiple-particle interference and quantum error correction},
  author={Steane, Andrew},
  journal={Proceedings of the Royal Society of London. Series A: Mathematical, Physical and Engineering Sciences},
  volume={452},
  number={1954},
  pages={2551--2577},
  year={1996},
  publisher={The Royal Society London},
  doi={10.1098/rspa.1996.0136
}
}

@article{blume2017demonstration,
  title={Demonstration of qubit operations below a rigorous fault tolerance threshold with gate set tomography},
  author={Blume-Kohout, Robin and Gamble, John King and Nielsen, Erik and Rudinger, Kenneth and Mizrahi, Jonathan and Fortier, Kevin and Maunz, Peter},
  journal={Nature communications},
  volume={8},
  number={1},
  pages={14485},
  year={2017},
  publisher={Nature Publishing Group UK London},
  doi={10.1038/ncomms14485}
}

@article{beale_efficiently_2021,
	title = {Efficiently computing logical noise in quantum error-correcting codes},
	volume = {103},
	number = {6},
	journal = {Physical Review A},
	author = {Beale, Stefanie J and Wallman, Joel J},
	year = {2021},
	note = {Publisher: APS},
	pages = {062404},
    doi = {10.1103/PhysRevA.103.062404}
}

@misc{combes2017logical,
      title={Logical Randomized Benchmarking}, 
      author={Joshua Combes and Christopher Granade and Christopher Ferrie and Steven T. Flammia},
      year={2017},
      eprint={1702.03688},
      archivePrefix={arXiv},
      primaryClass={quant-ph},
      url={https://arxiv.org/abs/1702.03688}, 
      doi = {10.48550/arXiv.1702.03688}
}

@article{modi2010role,
  title={Role of preparation in quantum process tomography},
  author={Modi, Kavan and Sudarshan, Ennackal Chandy George},
  journal={Physical Review A},
  volume={81},
  number={5},
  pages={052119},
  year={2010},
  publisher={APS},
  doi = {10.1103/PhysRevA.81.052119}
}

@misc{kobayashi2024tensor,
      title={Tensor-network decoders for process tensor descriptions of non-Markovian noise}, 
      author={Fumiyoshi Kobayashi and Hidetaka Manabe and Gregory A. L. White and Terry Farrelly and Kavan Modi and Thomas M. Stace},
      year={2024},
      eprint={2412.13739},
      archivePrefix={arXiv},
      primaryClass={quant-ph},
      url={https://arxiv.org/abs/2412.13739}, 
}

@book{levin2017markov,
  title={Markov chains and mixing times},
  author={Levin, David A and Peres, Yuval},
  volume={107},
  year={2017},
  publisher={American Mathematical Soc.},
  url = {https://pages.uoregon.edu/dlevin/MARKOV/markovmixing.pdf}
}

@book{nielsen2010quantum,
  title={Quantum computation and quantum information},
  author={Nielsen, Michael A and Chuang, Isaac L},
  year={2010},
  isbn = {1107002176},
  publisher = {Cambridge University Press},
  address = {USA},
  edition = {10th},
  nolink={}
}

@article{kam2024detrimental,
	author = {Kam, John Fidel and Gicev, Spiro and Modi, Kavan and Southwell, Angus and Usman, Muhammad},
	journal = {Quantum Science and Technology},
	title = {Detrimental non-Markovian errors for surface code memory},
	url = {http://iopscience.iop.org/article/10.1088/2058-9565/adebab},
	year = {2025}}

@article{bartlett_reference_2007,
	title = {Reference frames, superselection rules, and quantum information},
	volume = {79},
	number = {2},
	journal = {Reviews of Modern Physics},
	author = {Bartlett, Stephen D and Rudolph, Terry and Spekkens, Robert W},
	year = {2007},
	note = {Publisher: APS},
	pages = {555--609},
    doi= {10.1103/RevModPhys.79.555}
}

@misc{carrozza2024correspondence,
      title={A correspondence between quantum error correcting codes and quantum reference frames}, 
      author={Sylvain Carrozza and Aidan Chatwin-Davies and Philipp A. Hoehn and Fabio M. Mele},
      year={2025},
      eprint={2412.15317},
      archivePrefix={arXiv},
      primaryClass={quant-ph},
      url={https://arxiv.org/abs/2412.15317}, 
}

@article{calderbank1996good,
  title={Good quantum error-correcting codes exist},
  author={Calderbank, A Robert and Shor, Peter W},
  journal={Physical Review A},
  volume={54},
  number={2},
  pages={1098},
  year={1996},
  publisher={APS},
  doi={10.1103/PhysRevA.54.1098}
}

@phdthesis{gottesman1997stabilizer,
  title={Stabilizer Codes and Quantum Error Correction},
  author={Gottesman, Daniel Eric},
  year={1997},
  school={California Institute of Technology},
  doi ={10.48550/arXiv.quant-ph/9705052}
}

@misc{greenbaum_introduction_2015,
	title = {Introduction to {Quantum} {Gate} {Set} {Tomography}},
	url = {http://arxiv.org/abs/1509.02921},
	doi = {10.48550/arXiv.1509.02921},
	urldate = {2023-05-10},
	publisher = {arXiv},
	author = {Greenbaum, Daniel},
	month = sep,
	year = {2015},
    eprint={1509.02921},
    archivePrefix={arXiv},
    primaryClass={quant-ph},
	doi = {10.48550/arXiv.1509.02921},
}

@article{bravyi2018correcting,
  title={Correcting coherent errors with surface codes},
  author={Bravyi, Sergey and Englbrecht, Matthias and K{\"o}nig, Robert and Peard, Nolan},
  journal={npj Quantum Information},
  volume={4},
  number={1},
  pages={55},
  year={2018},
  publisher={Nature Publishing Group UK London},
  doi= {10.1038/s41534-018-0106-y}
}

@article{Nielsen2021gatesettomography,
  doi = {10.22331/q-2021-10-05-557},
  url = {https://doi.org/10.22331/q-2021-10-05-557},
  title = {Gate {S}et {T}omography},
  author = {Nielsen, Erik and Gamble, John King and Rudinger, Kenneth and Scholten, Travis and Young, Kevin and Blume-Kohout, Robin},
  journal = {{Quantum}},
  issn = {2521-327X},
  publisher = {{Verein zur F{\"{o}}rderung des Open Access Publizierens in den Quantenwissenschaften}},
  volume = {5},
  pages = {557},
  month = oct,
  year = {2021}
}

@article{li_concepts_2018,
	title = {Concepts of quantum non-{Markovianity}: {A} hierarchy},
	volume = {759},
	issn = {0370-1573},
	url = {https://www.sciencedirect.com/science/article/pii/S0370157318301601},
	doi = {https://doi.org/10.1016/j.physrep.2018.07.001},
	journal = {Physics Reports},
	author = {Li, Li and Hall, Michael J. W. and Wiseman, Howard M.},
	year = {2018},
	pages = {1--51},
}

@article{rudinger_probing_2019,
	title = {Probing {Context}-{Dependent} {Errors} in {Quantum} {Processors}},
	volume = {9},
	url = {https://link.aps.org/doi/10.1103/PhysRevX.9.021045},
	doi = {10.1103/PhysRevX.9.021045},
	number = {2},
	urldate = {2023-09-19},
	journal = {Physical Review X},
	author = {Rudinger, Kenneth and Proctor, Timothy and Langharst, Dylan and Sarovar, Mohan and Young, Kevin and Blume-Kohout, Robin},
	month = jun,
	year = {2019},
	note = {Publisher: American Physical Society},
	pages = {021045},
}

@inproceedings{rudinger2023probing,
  title={Probing logical error models with gate set tomography},
  author={Rudinger, Kenneth and Ziyad, Jalan and Morford-Oberst, Mario and Campos, Julie and Seritan, Stefan and Metodi, Tzvetan and Blume-Kohout, Robin},
  booktitle={APS March Meeting Abstracts},
  volume={2023},
  pages={D72--011},
  year={2023},
  url= {https://archive.aps.org/mar/2023/d72/11/}
}

@article{google_quantum_ai_suppressing_2023,
	title = {Suppressing quantum errors by scaling a surface code logical qubit},
	doi = {https://doi.org/10.1038/s41586-022-05434-1},
	number = {614},
	journal = {Nature},
	author = {Google Quantum AI},
	month = feb,
	year = {2023},
	pages = {676--681},
}

@article{shor1995scheme,
  title={Scheme for reducing decoherence in quantum computer memory},
  author={Shor, Peter W},
  journal={Physical review A},
  volume={52},
  number={4},
  pages={R2493},
  year={1995},
  publisher={APS},
  doi = {10.1103/PhysRevA.52.R2493},
}

@misc{blume2025quantum,
      title={Quantum Characterization, Verification, and Validation}, 
      author={Robin Blume-Kohout and Timothy Proctor and Kevin Young},
      year={2025},
      eprint={2503.16383},
      archivePrefix={arXiv},
      primaryClass={quant-ph},
      url={https://arxiv.org/abs/2503.16383}, 
}

@article{hashim2408practical,
  title = {Practical Introduction to Benchmarking and Characterization of Quantum Computers},
  author = {Hashim, Akel and Nguyen, Long B. and Goss, Noah and Marinelli, Brian and Naik, Ravi K. and Chistolini, Trevor and Hines, Jordan and Marceaux, J.P. and Kim, Yosep and Gokhale, Pranav and Tomesh, Teague and Chen, Senrui and Jiang, Liang and Ferracin, Samuele and Rudinger, Kenneth and Proctor, Timothy and Young, Kevin C. and Siddiqi, Irfan and Blume-Kohout, Robin},
  journal = {PRX Quantum},
  volume = {6},
  issue = {3},
  pages = {030202},
  numpages = {132},
  year = {2025},
  month = {Aug},
  publisher = {American Physical Society},
  doi = {10.1103/PRXQuantum.6.030202},
  url = {https://link.aps.org/doi/10.1103/PRXQuantum.6.030202}
}

@article{proctor2025benchmarking,
  title={Benchmarking quantum computers},
  author={Proctor, Timothy and Young, Kevin and Baczewski, Andrew D and Blume-Kohout, Robin},
  journal={Nature Reviews Physics },
  pages={1--14},
  year={2025},
  publisher={Nature Publishing Group UK London},
  doi ={10.1038/s42254-024-00796-z}
}

@article{laflamme1996perfect,
  title={Perfect quantum error correcting code},
  author={Laflamme, Raymond and Miquel, Cesar and Paz, Juan Pablo and Zurek, Wojciech Hubert},
  journal={Physical Review Letters},
  volume={77},
  number={1},
  pages={198},
  year={1996},
  publisher={APS},
  doi= {10.1103/PhysRevLett.77.198}
}

@article{peres1985reversible,
  title={Reversible logic and quantum computers},
  author={Peres, Asher},
  journal={Physical review A},
  volume={32},
  number={6},
  pages={3266},
  year={1985},
  publisher={APS},
  doi={10.1103/PhysRevA.32.3266}
}

@article{gottesman2024surviving,
  title={Surviving as a quantum computer in a classical world},
  author={Gottesman, Daniel},
  journal={Textbook manuscript preprint},
  year={2024},
  url= {https://www.cs.umd.edu/~dgottesm/QECCbook-2024.pdf}
}

@article{gorini1976completely,
  title={Completely positive dynamical semigroups of N-level systems},
  author={Gorini, Vittorio and Kossakowski, Andrzej and Sudarshan, Ennackal Chandy George},
  journal={Journal of Mathematical Physics},
  volume={17},
  number={5},
  pages={821--825},
  year={1976},
  publisher={American Institute of Physics},
  doi= {10.1063/1.522979}
}

@article{zwanzig1964identity,
  title={On the identity of three generalized master equations},
  author={Zwanzig, Robert},
  journal={Physica},
  volume={30},
  number={6},
  pages={1109--1123},
  year={1964},
  publisher={Elsevier},
  doi ={10.1016/0031-8914(64)90102-8}
}

@ARTICLE{Wolf2008-hx,
  title     = "Dividing Quantum Channels",
  author    = "Wolf, Michael M and Cirac, J Ignacio",
  journal   = "Commun. Math. Phys.",
  publisher = "Springer",
  volume    =  279,
  number    =  1,
  pages     = "147--168",
  month     =  apr,
  year      =  2008,
  doi       = "10.1007/s00220-008-0411-y",
  issn      = "1432-0916"
}

@ARTICLE{Wolf2008-mi,
  title     = "Assessing non-Markovian quantum dynamics",
  author    = "Wolf, M M and Eisert, J and Cubitt, T S and Cirac, J I",
  journal   = "Phys. Rev. Lett.",
  publisher = "APS",
  volume    =  101,
  number    =  15,
  pages     =  150402,
  month     =  oct,
  year      =  2008,
  doi       = "10.1103/PhysRevLett.101.150402",
  pmid      =  18999575,
  issn      = "0031-9007",
  language  = "en"
}

@article{yoder2016universal,
  title={Universal fault-tolerant gates on concatenated stabilizer codes},
  author={Yoder, Theodore J and Takagi, Ryuji and Chuang, Isaac L},
  journal={Physical Review X},
  volume={6},
  number={3},
  pages={031039},
  year={2016},
  publisher={APS},
  doi={10.1103/PhysRevX.6.031039}
}

@ARTICLE{Rivas2014-review,
  title     = "Quantum non-Markovianity: characterization, quantification and
               detection",
  author    = "Rivas, Angel and Huelga, Susana F and Plenio, Martin B",
  journal   = "Rep. Prog. Phys.",
  publisher = "iopscience.iop.org",
  volume    =  77,
  number    =  9,
  pages     =  094001,
  month     =  sep,
  year      =  2014,
  doi       = "10.1088/0034-4885/77/9/094001",
  pmid      =  25147025,
  issn      = "1361-6633,0034-4885"
}

@ARTICLE{Breuer2016-Colloquium,
  title     = "Colloquium: Non-Markovian dynamics in open quantum systems",
  author    = "Breuer, Heinz-Peter and Laine, Elsi-Mari and Piilo, Jyrki and
               Vacchini, Bassano",
  journal   = "Rev. Mod. Phys.",
  publisher = "American Physical Society",
  volume    =  88,
  number    =  2,
  pages     =  021002,
  month     =  apr,
  year      =  2016,
  doi       = "10.1103/RevModPhys.88.021002",
  issn      = "0034-6861"
}

@article{breuer2009measure,
  title={Measure for the degree of non-Markovian behavior of quantum processes in open systems},
  author={Breuer, Heinz-Peter and Laine, Elsi-Mari and Piilo, Jyrki},
  journal={Physical review letters},
  volume={103},
  number={21},
  pages={210401},
  year={2009},
  publisher={APS},
  doi={10.1103/PhysRevLett.103.210401}
}

@ARTICLE{Pollock2018-operational,
  title    = "Operational Markov Condition for Quantum Processes",
  author   = "Pollock, Felix A and Rodríguez-Rosario, César and Frauenheim,
              Thomas and Paternostro, Mauro and Modi, Kavan",
  journal  = "Phys. Rev. Lett.",
  volume   =  120,
  number   =  4,
  pages    =  040405,
  month    =  jan,
  year     =  2018,
  doi      = "10.1103/PhysRevLett.120.040405",
  pmid     =  29437441,
  issn     = "0031-9007,1079-7114",
  language = "en"
}

@ARTICLE{White2022-NMQPT,
  title     = "Non-Markovian Quantum Process Tomography",
  author    = "White, G A L and Pollock, F A and Hollenberg, L C L and Modi, K
               and Hill, C D",
  journal   = "PRX Quantum",
  publisher = "American Physical Society",
  volume    =  3,
  number    =  2,
  pages     =  020344,
  month     =  may,
  year      =  2022,
  doi       = "10.1103/PRXQuantum.3.020344"
}

@article{vezvaee2026surface,
  title={Surface code scaling on heavy-hex superconducting quantum processors},
  author={Vezvaee, Arian and Benito, Cesar and Morford-Oberst, Mario and Bermudez, Alejandro and Lidar, Daniel A},
  journal={Nature Communications},
  year={2026},
  publisher={Nature Publishing Group UK London},
  doi={10.1038/s41467-026-76090-6}
}

@article{rivas2010entanglement,
  title={Entanglement and non-Markovianity of quantum evolutions},
  author={Rivas, {\'A}ngel and Huelga, Susana F and Plenio, Martin B},
  journal={Physical review letters},
  volume={105},
  number={5},
  pages={050403},
  year={2010},
  publisher={APS},
  doi = {10.1103/PhysRevLett.105.050403},
  url = {https://link.aps.org/doi/10.1103/PhysRevLett.105.050403}
}

@article{heussen2024measurement,
  title = {Measurement-Free Fault-Tolerant Quantum Error Correction in Near-Term Devices},
  author = {Heu\ss{}en, Sascha and Locher, David F. and M\"uller, Markus},
  journal = {PRX Quantum},
  volume = {5},
  issue = {1},
  pages = {010333},
  numpages = {25},
  year = {2024},
  month = {Feb},
  publisher = {American Physical Society},
  doi = {10.1103/PRXQuantum.5.010333},
  url = {https://link.aps.org/doi/10.1103/PRXQuantum.5.010333}
}

@article{butt_demonstration_2026,
	title = {Demonstration of measurement-free universal logical quantum computation},
	volume = {17},
	issn = {2041-1723},
	url = {https://doi.org/10.1038/s41467-026-68533-x},
	doi = {10.1038/s41467-026-68533-x},
	number = {1},
	journal = {Nature Communications},
	author = {Butt, Friederike and Pogorelov, Ivan and Freund, Robert and Steiner, Alex and Meyer, Marcel and Monz, Thomas and Müller, Markus},
	month = jan,
	year = {2026},
	pages = {995},
}

@article{agarwal2024modelling,
  title={Modelling non-Markovian noise in driven superconducting qubits},
  author={Agarwal, Abhishek and Lindoy, Lachlan P and Lall, Deep and Jamet, Fran{\c{c}}ois and Rungger, Ivan},
  journal={Quantum Science and Technology},
  volume={9},
  number={3},
  pages={035017},
  year={2024},
  publisher={IOP Publishing},
  doi = {10.1088/2058-9565/ad3d7e}
}

@article{zhuang2026non,
  title={Non-Markovian relaxation spectroscopy of fluxonium qubits},
  author={Zhuang, Ze-Tong and Rosenstock, Dario and Liu, Bao-Jie and Somoroff, Aaron and Manucharyan, Vladimir E and Wang, Chen},
  journal={Nature Communications},
  volume={17},
  number={1},
  pages={3209},
  year={2026},
  publisher={Nature Publishing Group UK London},
doi = {10.1038/s41467-026-69910-2}
}

@article{hayes_eliminating_2020,
	title = {Eliminating {Leakage} {Errors} in {Hyperfine} {Qubits}},
	volume = {124},
	issn = {0031-9007, 1079-7114},
	url = {https://link.aps.org/doi/10.1103/PhysRevLett.124.170501},
	doi = {10.1103/PhysRevLett.124.170501},
	language = {en},
	number = {17},
	urldate = {2026-08-28},
	journal = {Physical Review Letters},
	author = {Hayes, D. and Stack, D. and Bjork, B. and Potter, A. C. and Baldwin, C. H. and Stutz, R. P.},
	month = apr,
	year = {2020},
	pages = {170501},
}

@article{ostrove_praqtice_2025,
	title = {{PRAQTICE}},
	url = {https://www.osti.gov/doecode/biblio/158181},
	doi = {10.11578/DC.20250721.1},
	urldate = {2026-11-23},
	publisher = {Sandia National Laboratories (SNL-CA), Livermore, CA (United States)},
    journal ={https://www.osti.gov/doecode/biblio/158181},
	author = {Ostrove, Corey and Proctor, Timothy and Hines, Jordan and Rudinger, Kenneth and Seritan, Stefan},
	year = {2025},

}

@article{babu_state_2021,
	title = {State leakage during fast decay and control of a superconducting transmon qubit},
	volume = {7},
	issn = {2056-6387},
	url = {https://www.nature.com/articles/s41534-020-00357-z},
	doi = {10.1038/s41534-020-00357-z},
	language = {en},
	number = {1},
	urldate = {2026-08-28},
	journal = {npj Quantum Information},
	author = {Babu, Aravind Plathanam and Tuorila, Jani and Ala-Nissila, Tapio},
	month = feb,
	year = {2021},
	pages = {30},
}

@article{yoshihara_decoherence_2006,
	title = {Decoherence of {Flux} {Qubits} due to 1 / f {Flux} {Noise}},
	volume = {97},
	copyright = {http://link.aps.org/licenses/aps-default-license},
	issn = {0031-9007, 1079-7114},
	url = {https://link.aps.org/doi/10.1103/PhysRevLett.97.167001},
	doi = {10.1103/PhysRevLett.97.167001},
	language = {en},
	number = {16},
	urldate = {2026-08-28},
	journal = {Physical Review Letters},
	author = {Yoshihara, F. and Harrabi, K. and Niskanen, A. O. and Nakamura, Y. and Tsai, J. S.},
	month = oct,
	year = {2006},
	pages = {167001},
}

@article{slichter_measurement-induced_2012,
	title = {Measurement-{Induced} {Qubit} {State} {Mixing} in {Circuit} {QED} from {Up}-{Converted} {Dephasing} {Noise}},
	volume = {109},
	copyright = {http://link.aps.org/licenses/aps-default-license},
	issn = {0031-9007, 1079-7114},
	url = {https://link.aps.org/doi/10.1103/PhysRevLett.109.153601},
	doi = {10.1103/PhysRevLett.109.153601},
	language = {en},
	number = {15},
	urldate = {2026-08-28},
	journal = {Physical Review Letters},
	author = {Slichter, D. H. and Vijay, R. and Weber, S. J. and Boutin, S. and Boissonneault, M. and Gambetta, J. M. and Blais, A. and Siddiqi, I.},
	month = oct,
	year = {2012},
	pages = {153601},
}

@article{wellstood_low-frequency_1987,
	title = {Low-frequency noise in dc superconducting quantum interference devices below 1 {K}},
	volume = {50},
	issn = {0003-6951, 1077-3118},
	url = {https://pubs.aip.org/apl/article/50/12/772/52092/Low-frequency-noise-in-dc-superconducting-quantum},
	doi = {10.1063/1.98041},
	language = {en},
	number = {12},
	urldate = {2026-08-28},
	journal = {Applied Physics Letters},
	author = {Wellstood, Frederick C. and Urbina, Cristian and Clarke, John},
	month = mar,
	year = {1987},
	pages = {772--774},
}

@article{medford_self-consistent_2013,
	title = {Self-consistent measurement and state tomography of an exchange-only spin qubit},
	volume = {8},
	copyright = {http://www.springer.com/tdm},
	issn = {1748-3387, 1748-3395},
	url = {https://www.nature.com/articles/nnano.2013.168},
	doi = {10.1038/nnano.2013.168},
	language = {en},
	number = {9},
	urldate = {2026-08-28},
	journal = {Nature Nanotechnology},
	author = {Medford, J. and Beil, J. and Taylor, J. M. and Bartlett, S. D. and Doherty, A. C. and Rashba, E. I. and DiVincenzo, D. P. and Lu, H. and Gossard, A. C. and Marcus, C. M.},
	month = sep,
	year = {2013},
	pages = {654--659},
}

@article{lucero_high-fidelity_2008,
	title = {High-{Fidelity} {Gates} in a {Single} {Josephson} {Qubit}},
	volume = {100},
	copyright = {http://link.aps.org/licenses/aps-default-license},
	issn = {0031-9007, 1079-7114},
	url = {https://link.aps.org/doi/10.1103/PhysRevLett.100.247001},
	doi = {10.1103/PhysRevLett.100.247001},
	language = {en},
	number = {24},
	urldate = {2026-08-28},
	journal = {Physical Review Letters},
	author = {Lucero, Erik and Hofheinz, M. and Ansmann, M. and Bialczak, Radoslaw C. and Katz, N. and Neeley, Matthew and O’Connell, A. D. and Wang, H. and Cleland, A. N. and Martinis, John M.},
	month = jun,
	year = {2008},
	pages = {247001},
}

@misc{endo2025non,
      title={Non-Markovianity in Quantum Information Processing: Interplay with Quantum Error Mitigation}, 
      author={Suguru Endo and Hideaki Hakoshima and Tomohiro Shitara},
      year={2025},
      eprint={2510.20224},
      archivePrefix={arXiv},
      primaryClass={quant-ph},
      url={https://arxiv.org/abs/2510.20224}, 
}
\end{document}